\documentclass[12pt,english]{article}
\usepackage[T1]{fontenc}
\usepackage[utf8]{inputenc}
\usepackage{babel}
\usepackage{float}
\usepackage{mathtools}
\usepackage{enumitem}
\usepackage{tikz}
\usepackage{amsmath}
\usepackage{amsthm}
\usepackage{geometry}
\usepackage{setspace}
\usepackage[]
 {hyperref}

\makeatletter

\providecolor{lyxadded}{rgb}{0,0,1}
\providecolor{lyxdeleted}{rgb}{1,0,0}
\usetikzlibrary{calc}
\newcommand{\lyxobjectsout}[1]{%
  \bgroup%
  \color{lyxdeleted}%
  \tikz{
    \node[inner sep=0pt,outer sep=0pt](lyxdelobj){#1};
    \draw($(lyxdelobj.south west)+(2em,.5em)$)--($(lyxdelobj.north east)-(2em,.5em)$);
  }
  \egroup%
}
\DeclareRobustCommand{\lyxdisplayobjdeleted}[4][]{%
  \ifx#4\empty\else%
     \texorpdfstring{\leavevmode\\\lyxobjectsout{\parbox{\linewidth}{#4}}}{}%
  \fi%
}
\DeclareRobustCommand{\lyxudisplayobjdeleted}[4][]{%
  \ifx#4\empty\else%
     \texorpdfstring{\leavevmode\\\raisebox{-\belowdisplayshortskip}{%
                \lyxobjectsout{\parbox[b]{\linewidth}{#4}}}}{}%
     \leavevmode\\%
  \fi%
}

\theoremstyle{plain}
\newtheorem{thm}{\protect\theoremname}
\theoremstyle{plain}
\newtheorem{cor}{\protect\corollaryname}
\theoremstyle{remark}
\newtheorem{rem}{\protect\remarkname}
\theoremstyle{plain}
\newtheorem{lem}{\protect\lemmaname}

\usepackage{bm}                  % 支持加粗数学符号，如 \bm{x}
\usepackage{pdflscape}
\usepackage[bottom]{footmisc}
\usepackage{microtype} %margin
\usepackage{enumitem}
\setlist[enumerate]{label=(\roman*), itemsep=0pt}

\usepackage{indentfirst}

\renewcommand{\tilde}[1]{\widetilde{#1}}

\usepackage{siunitx}            % 数值与单位排版
\usepackage{booktabs}
\usepackage{threeparttable}
\usepackage{caption}
\usepackage{hyperref}
\hypersetup{
  colorlinks=true,
  linkcolor=red,
  citecolor=[RGB]{0,0,139},
  urlcolor=black
}
\newcommand{\xrefblack}[1]{\hyperref[#1]{\textcolor{black}{\ref*{#1}}}}

\renewenvironment{proof}[1][\proofname]{\par
    \pushQED{\qed}%
    \normalfont \topsep12\p@\@plus4\p@\@minus3\p@\relax
    \trivlist
    \item\relax
          {\bfseries
      #1\@addpunct{.}}\hspace\labelsep\ignorespaces
  }{%
    \popQED\endtrivlist\@endpefalse
    \addvspace{13pt plus 4pt minus 3pt}
  }

\def\thm@space@setup{%
  \thm@preskip=12pt plus 4pt minus 3pt
  \thm@postskip=1.05\thm@preskip
}

\def\th@remark{%
  \thm@headfont{\bfseries}
  \normalfont
  \thm@preskip=12pt plus 4pt minus 3pt
  \thm@postskip\thm@preskip
}

\allowdisplaybreaks % 允许公式跨页

\usepackage{ragged2e} % gives \justifying

\usepackage{titlesec}
\titlelabel{\thetitle.\quad} % 此命令为所有标题标签（chapter, section等）添加句点
\titleformat{\chapter}[display]{\normalfont\huge\bfseries}{\chaptertitlename\ \thechapter.}{20pt}{\Huge}
\titleformat{\section}{\normalfont\Large\bfseries}{\thesection.}{0.75em}{}
\titleformat{\subsection}{\normalfont\large\bfseries}{\thesubsection.}{0.75em}{}

\usepackage{biblatex-econ}  

\makeatother

\usepackage{biblatex}
\providecommand{\corollaryname}{Corollary}
\providecommand{\lemmaname}{Lemma}
\providecommand{\remarkname}{Remark}
\providecommand{\theoremname}{Theorem}

\begin{document}
\global\long\def\a{\alpha}%
\global\long\def\b{\beta}%
\global\long\def\g{\gamma}%
\global\long\def\d{\delta}%
\global\long\def\e{\epsilon}%
\global\long\def\l{\lambda}%
\global\long\def\t{\theta}%
\global\long\def\o{\omega}%
\global\long\def\s{\sigma}%
\global\long\def\G{\Gamma}%
\global\long\def\D{\Delta}%
\global\long\def\L{\Lambda}%
\global\long\def\T{\Theta}%
\global\long\def\O{\Omega}%
\global\long\def\R{\mathbb{R}}%
\global\long\def\N{\mathbb{N}}%
\global\long\def\H{\mathbb{H}}%
\global\long\def\Q{\mathbb{Q}}%
\global\long\def\I{\mathbb{I}}%
\global\long\def\P{P}%
\global\long\def\E{E}%
\global\long\def\B{\mathbb{\mathbb{B}}}%
\global\long\def\S{\mathbb{\mathbb{S}}}%
\global\long\def\V{\mathbb{\mathbb{V}}\text{ar}}%
 
\global\long\def\cE{\mathbb{\mathcal{E}}}%
 
\global\long\def\X{{\bf X}}%
\global\long\def\cX{\mathscr{X}}%
\global\long\def\cY{\mathscr{Y}}%
\global\long\def\cA{\mathscr{A}}%
\global\long\def\cB{\mathscr{B}}%
\global\long\def\cM{\mathscr{M}}%
\global\long\def\cN{\mathcal{N}}%
\global\long\def\cG{\mathcal{G}}%
\global\long\def\cC{\mathcal{C}}%
\global\long\def\sp{\,}%
\global\long\def\es{\emptyset}%
\global\long\def\mc#1{\mathscr{#1}}%
\global\long\def\ind{\mathbf{\mathbbm1}}%
\global\long\def\indep{\perp}%
\global\long\def\any{\forall}%
\global\long\def\ex{\exists}%
\global\long\def\p{\partial}%
\global\long\def\cd{\cdot}%
\global\long\def\Dif{\nabla}%
\global\long\def\imp{\Rightarrow}%
\global\long\def\iff{\Leftrightarrow}%
\global\long\def\up{\uparrow}%
\global\long\def\down{\downarrow}%
\global\long\def\arrow{\rightarrow}%
\global\long\def\rlarrow{\leftrightarrow}%
\global\long\def\lrarrow{\leftrightarrow}%
\global\long\def\abs#1{\left|#1\right|}%
\global\long\def\norm#1{\left\Vert #1\right\Vert }%
\global\long\def\rest#1{\left.#1\right|}%
\global\long\def\bracket#1#2{\left\langle #1\middle\vert#2\right\rangle }%
\global\long\def\sandvich#1#2#3{\left\langle #1\middle\vert#2\middle\vert#3\right\rangle }%
\global\long\def\turd#1{\frac{#1}{3}}%
\global\long\def\ellipsis{\textellipsis}%
\global\long\def\sand#1{\left\lceil #1\right\vert }%
\global\long\def\wich#1{\left\vert #1\right\rfloor }%
\global\long\def\sandwich#1#2#3{\left\lceil #1\middle\vert#2\middle\vert#3\right\rfloor }%
\global\long\def\abs#1{\left|#1\right|}%
\global\long\def\norm#1{\left\Vert #1\right\Vert }%
\global\long\def\rest#1{\left.#1\right|}%
\global\long\def\inprod#1{\left\langle #1\right\rangle }%
\global\long\def\ol#1{\overline{#1}}%
\global\long\def\ul#1{\underline{#1}}%
\global\long\def\td#1{\tilde{#1}}%
\global\long\def\bs#1{\boldsymbol{#1}}%
\global\long\def\upto{\nearrow}%
\global\long\def\downto{\searrow}%
\global\long\def\pto{\stackrel{p}{\to}}%
\global\long\def\dto{\stackrel{d}{\to}}%
\global\long\def\asto{\rightarrow_{a.s.}}%
\global\long\def\gto{\rightarrow}%
\global\long\def\fto{\Rightarrow}%
\global\long\def\Tr{\mathrm{Tr}}%
\global\long\def\cov{\mathrm{cov}}%
\global\long\def\var{\mathrm{var}}%
\global\long\def\plim{\mathrm{plim}}%
\global\long\def\as{\mathrm{a.s.}}%
\global\long\def\diag{\mathrm{diag}}%

\title{Initial-Condition-Robust Inference \\in Autoregressive Models}
\author{Donald W. K. Andrews\thanks{Andrews: Department of Economics and Cowles Foundation, Yale University,
\protect\href{mailto:donald.andrews@yale.edu}{donald.andrews@yale.edu}.} \and Ming Li\thanks{Li: Department of Economics and Risk Management Institute, National
University of Singapore, \protect\href{mailto:mli@nus.edu.sg}{mli@nus.edu.sg}.} \and Yapeng Zheng\thanks{Zheng: Department of Economics, Chinese University of Hong Kong, \protect\href{http://YapengZheng@link.cuhk.edu.hk}{yapengzheng@link.cuhk.edu.hk}.}}
\maketitle
\begin{abstract}
This paper considers confidence intervals (CIs) for the autoregressive
(AR) parameter in an AR model with an AR parameter that may be close
or equal to one. Existing CIs rely on the assumption of a stationary
or fixed initial condition to obtain correct asymptotic coverage and
good finite sample coverage. When this assumption fails, their coverage
can be quite poor. In this paper, we introduce a new CI for the AR
parameter whose coverage probability is completely robust to the initial
condition, both asymptotically and in finite samples. This CI pays
only a small price in terms of its length when the initial condition
is stationary or fixed. The new CI also is robust to conditional heteroskedasticity
of the errors.

\vspace{0.4in}

 \medskip

\end{abstract}
\thispagestyle{empty}

\newpage{}

\setcounter{page}{1}

\section{Introduction}

We consider an AR model with an AR parameter that may take values
close or equal to one. Existing CIs in the literature for the AR parameter
assume that the initial condition is stationary, zero, or fixed, e.g.,
see \textcite{stock1991confidence}, \textcite{andrews1993exactly},
\textcite{andrews1994approximately}, \textcite{hansen1999grid},
\textcite{elliott2001confidence}, \textcite{mikusheva2007uniform},
and \textcite[AG14 hereafter]{andrews2014conditional}. If the initial-condition
assumption is not satisfied, then these CIs do not have correct asymptotic
coverage probabilities and their finite sample coverage probabilities
can be poor. For example, simulations reported below show that the
nominal 95\% AG14 CI has finite sample coverage probabilities ranging
from 24.1\% to 93.5\% across 50 cases with highly variable initial
conditions and different types of conditional heteroskedasticity of
the error. A quarter of these CPs are 79.0\% or less. 

To circumvent the sensitivity of existing CIs to the initial condition,
this paper introduces a new initial-condition-robust (ICR) CI whose
finite sample CPs do not depend on the initial condition. These CIs
are easy to compute and do not depend on any tuning parameters. We
show that the ICR CI has correct asymptotic size in a uniform sense
under conditions that allow for an arbitrary initial condition and
conditional heteroskedasticity. Simulations show that the finite sample
CPs of the nominal 95\% ICR CI are quite good. They lie between 93.5\%
and 95.0\% in the cases considered. 

Like the other CIs referenced above, the ICR CI is constructed by
inverting hypothesis tests concerning the value of the AR parameter.
For all of these CIs, the tests are based on a t statistic that depends
on a least squares (LS) estimator of the AR parameter. In contrast
to existing CIs, the LS estimator used by the ICR tests contains an
additional regressor that eliminates the effect of the initial condition
under the null hypothesis. (For the definition of this regressor,
see Section \ref{sec:Model-and-ICR} below.) In consequence, the ICR
CI has a CP that is invariant to the value of the initial condition.
However, the length of the ICR CI can be effected by the initial condition.
Simulations show that the ICR CI is slightly shorter when the initial
condition is highly variable than when it is stationary or fixed.\footnote{The length of the ICR CI varies mostly with the value of the AR parameter.
It is shortest for AR parameters near one. This is true of existing
CIs as well.}

In scenarios where existing CIs are asymptotically valid (i.e., the
initial condition is stationary or fixed), the ICR CI pays a price
in terms of its length compared to existing CIs. This occurs because
the ICR LS estimator has a higher variance due to the additional regressor
that is included in the LS regression. However, simulations show that
the effect is fairly small. Across 50 scenarios with stationary or
fixed initial conditions and i.i.d., GARCH, or ARCH errors, the ratio
of the expected length of the ICR CI compared to that of the AG14
CI is found to vary between 1.00 and 1.11. On average, the ICR CI
is 3.5\% longer than the AG14 CI in these scenarios.

Based on the ICR CI, we also introduce an asymptotically median-unbiased
interval estimator (MUE) of the AR parameter.

The AR model considered here has been applied in the literature to
exchange rate, commodity and stock prices, and other economic time
series, e.g., see \textcite{kim1993unit}.

This paper is organized as follows. Section \ref{sec:Model-and-ICR}
specifies the AR model and defines the ICR CI. Section \ref{sec:MUE}
defines the median-unbiased interval estimator. Section \ref{sec:Monte-Carlo-Simulations}
provides simulation results concerning CPs and average lengths of
the ICR and AG14 CIs and absolute median biases of the ICR\ MUE.
Section \ref{sec:Asymptotic-Results} provides asymptotic results
for the ICR CI. The Supplemental Material describes how the critical
values were computed, proves Theorems \ref{thm:ICR-CP} and \ref{thm:Jh-convergence},
and provides some additional simulation results. 

\section{Model and ICR Confidence Interval}\label{sec:Model-and-ICR}

We consider the AR(1) model with conditional-heteroskedasticity studied
in AG14:
\begin{align}
Y_{i} & =\mu+Y_{i}^{*}\quad\text{for \ensuremath{i=0,1,\dots,n} and}\nonumber \\
Y_{i}^{*} & =\rho Y_{i-1}^{*}+U_{i}\quad\text{for }i=1,2,\dots,n,\label{eq:model}
\end{align}
where the observed data are $\{Y_{0},\dots,Y_{n}\}$, $\rho\in[-1+\varepsilon,1]$
for some $0<\varepsilon<2$, and $\{U_{i}:i=1,2,\dots\}$ are stationary
and ergodic under the distribution $F$, with conditional mean $0$
given a $\sigma$-field $\mathcal{G}_{i-1}$ for which $U_{j}\in\mathcal{G}_{i}$
for all $j\leq i$, conditional variance $\sigma_{i}^{2}=\E_{F}(U_{i}^{2}|\mathcal{G}_{i-1})$,
and unconditional variance $\sigma_{U}^{2}\in(0,\infty)$. A detailed
description of the parameter space $\L$ is provided in Section \ref{subsec:Parameter-space}.

By iterative substitution in (\ref{eq:model}), we have 
\begin{align}
Y_{i} & =\mu+\rho^{i-1}\cdot\rho\cdot Y_{0}^{*}+\sum_{j=1}^{i}\rho^{i-j}U_{j}\quad\text{for }i=1,\dots,n.\label{eq:representation-ar1}
\end{align}
 Let $Y$, $U$, and $X_{1}$ be $n$-vectors whose $i$-th elements
are $Y_{i}$, $U_{i}$, and $Y_{i-1}$, respectively. Let $X_{2}(\rho)$
be an $n\times2$ matrix whose $i$-th row is $(1,\rho^{i-1})$ when
$-1<\rho<1$.\footnote{Following the convention $0^{0}=1$, we use $\rho^{i-1}$ in the regression
instead of $\rho^{i}$ to avoid perfect multicollinearity when $\rho=0$.} When $\rho=1$, we define $X_{2}(\rho)$ to be an $n\times2$ matrix
whose $i$-th row is $(1,i)$. Let $X(\rho)=[X_{1}:X_{2}(\rho)]$,
$P_{X(\rho)}=X(\rho)(X(\rho)'X(\rho))^{-1}X(\rho)'$, and $M_{X(\rho)}=I_{n}-P_{X(\rho)}$.
Let $\widehat{U}_{i}$ denote the $i$-th element of the residual
vector $M_{X(\rho)}Y$. Let $p_{ii}$ denote the $i$-th diagonal
element of $P_{X(\rho)}$, and define $p_{ii}^{*}=\min\{p_{ii},n^{-1/2}\}$.
Let $\Delta$ be a diagonal $n\times n$ matrix whose $i$-th diagonal
element is $\widehat{U}_{i}/(1-p_{ii}^{*})$. Then, the ICR LS estimator
$\widehat{\rho}_{n}(\rho)$, and the HC5 variance estimator $\widehat{\sigma}_{n}^{2}(\rho)$
are defined as 
\begin{align}
\widehat{\rho}_{n}(\rho)= & \ (X_{1}'M_{X_{2}(\rho)}X_{1})^{-1}X_{1}'M_{X_{2}(\rho)}Y\text{ and}\nonumber \\
\widehat{\sigma}_{n}^{2}(\rho)= & \ (n^{-1}X_{1}'M_{X_{2}(\rho)}X_{1})^{-1}(n^{-1}X_{1}'M_{X_{2}(\rho)}\Delta^{2}M_{X_{2}(\rho)}X_{1})(n^{-1}X_{1}'M_{X_{2}(\rho)}X_{1})^{-1},\label{eq:ICR_sigma_hat}
\end{align}
respectively.

Let 
\begin{equation}
T_{n}(\rho)=\frac{n^{1/2}(\widehat{\rho}_{n}(\rho)-\rho)}{\widehat{\sigma}_{n}^{2}(\rho)}.\label{eq:t-stat(rho)}
\end{equation}
For suitable sequences $\{\rho_{n}\}_{n\geq1}$ such that $n(1-\rho_{n})\to h\in[0,\infty]$,
we show that 
\begin{equation}
T_{n}(\rho_{n})\dto J_{h},\label{eq:asym-distribution}
\end{equation}
where $J_{h}$ is defined in Section \ref{subsec:Asymptotic-Distribution-Jh}
below. Table \ref{tab:ICR-quantiles} provides the quantiles $c_{h}(\alpha/2)$
and $c_{h}(1-\alpha/2)$ of the distribution of $J_{h}$ for $\alpha=.05$
and $.1$, which are the critical values employed by the ICR CI.

The nominal $1-\alpha$ equal-tailed two-sided ICR CI for $\rho$
is 
\begin{equation}
CI_{\text{ICR},n}:=\left\{ \rho\in[-1+\epsilon,1]:c_{h}(\alpha/2)\leq T_{n}(\rho)\leq c_{h}(1-\alpha/2)\ \text{for}\ h=n(1-\rho)\right\} ,\label{eq:ICR_CI}
\end{equation}
which can be computed by taking a fine grid of values $\rho\in[-1+\epsilon,1]$
and comparing $T_{n}(\rho)$ to $c_{h}(\alpha/2)$ and $c_{h}(1-\alpha/2)$.
Given these critical values, computation of the CIs is fast. 

\begin{table}[t]
\centering
\caption{Quantiles of $J_{h}$ for Use with 90\% and 95\% Equal-Tailed Two-Sided
CIs and MUEs}\label{tab:ICR-quantiles}

\centering
\scriptsize
\begin{tabular}{>{\centering\arraybackslash}m{.98cm}ccccccccccccc}
\toprule 
\multicolumn{14}{c}{{Values of $c_{h}\left(\alpha\right)$, the $\alpha^{\text{th}}$ Quantile of $J_{h}$, for Use with 90\% and 95\% Equal-Tailed Two-Sided CI's and MUE's}}\tabularnewline
\midrule
{$h$} & {0} & {.2} & {.4} & {.6} & {.8} & {1} & {1.4} & {1.8} & {2.2} & {2.6} & {3} & {3.4} & {3.8}\tabularnewline
\midrule
{$c_{h}\left(.025\right)$} & {-3.66} & {-3.63} & {-3.60} & {-3.56} & {-3.54} & {-3.52} & {-3.46} & {-3.40} & {-3.36} & {-3.31} & {-3.27} & {-3.23} & {-3.19}\tabularnewline
{$c_{h}\left(.05\right)$} & {-3.41} & {-3.38} & {-3.35} & {-3.31} & {-3.28} & {-3.25} & {-3.20} & {-3.14} & {-3.08} & {-3.04} & {-3.00} & {-2.95} & {-2.91}\tabularnewline
{$c_{h}\left(.5\right)$} & {-2.18} & {-2.13} & {-2.09} & {-2.04} & {-1.99} & {-1.95} & {-1.86} & {-1.78} & {-1.70} & {-1.63} & {-1.57} & {-1.50} & {-1.45}\tabularnewline
{$c_{h}\left(.95\right)$} & {-.94} & {-.87} & {-.80} & {-.74} & {-.68} & {-.62} & {-.50} & {-.39} & {-.29} & {-.19} & {-.11} & {-.03} & {.05}\tabularnewline
{$c_{h}\left(.975\right)$} & {-.65} & {-.59} & {-.52} & {-.45} & {-.38} & {-.32} & {-.21} & {-.08} & {.01} & {.11} & {.19} & {.28} & {.35}\tabularnewline
\midrule
{$h$} & {4.2} & {4.6} & {5} & {6} & {7} & {8} & {9} & {10} & {11} & {12} & {13} & {14} & {15}\tabularnewline
\midrule
{$c_{h}\left(.025\right)$} & {-3.16} & {-3.12} & {-3.09} & {-3.02} & {-2.97} & {-2.90} & {-2.87} & {-2.82} & {-2.79} & {-2.75} & {-2.73} & {-2.71} & {-2.69}\tabularnewline
{$c_{h}\left(.05\right)$} & {-2.87} & {-2.83} & {-2.80} & {-2.72} & {-2.66} & {-2.61} & {-2.56} & {-2.52} & {-2.48} & {-2.45} & {-2.42} & {-2.40} & {-2.38}\tabularnewline
{$c_{h}\left(.5\right)$} & {-1.39} & {-1.34} & {-1.30} & {-1.20} & {-1.11} & {-1.04} & {-.99} & {-.93} & {-.89} & {-.85} & {-.82} & {-.78} & {-.76}\tabularnewline
{$c_{h}\left(.95\right)$} & {.11} & {.18} & {.24} & {.36} & {.46} & {.55} & {.61} & {.68} & {.74} & {.78} & {.81} & {.84} & {.88}\tabularnewline
{$c_{h}\left(.975\right)$} & {.41} & {.48} & {.54} & {.66} & {.77} & {.86} & {.92} & {.99} & {1.05} & {1.08} & {1.12} & {1.15} & {1.20}\tabularnewline
\midrule
{$h$} & {20} & {25} & {30} & {40} & {50} & {60} & {70} & {80} & {90} & {100} & {200} & {300} & {500}\tabularnewline
\midrule
{$c_{h}\left(.025\right)$} & {-2.59} & {-2.53} & {-2.47} & {-2.41} & {-2.36} & {-2.32} & {-2.30} & {-2.27} & {-2.26} & {-2.25} & {-2.16} & {-2.13} & {-2.09}\tabularnewline
{$c_{h}\left(.05\right)$} & {-2.28} & {-2.22} & {-2.15} & {-2.09} & {-2.05} & {-2.01} & {-1.99} & {-1.96} & {-1.94} & {-1.94} & {-1.84} & {-1.81} & {-1.78}\tabularnewline
{$c_{h}\left(.5\right)$} & {-.65} & {-.58} & {-.52} & {-.45} & {-.41} & {-.37} & {-.34} & {-.32} & {-.30} & {-.28} & {-.20} & {-.16} & {-.13}\tabularnewline
{$c_{h}\left(.95\right)$} & {.99} & {1.06} & {1.12} & {1.19} & {1.24} & {1.28} & {1.30} & {1.32} & {1.34} & {1.36} & {1.45} & {1.48} & {1.52}\tabularnewline
{$c_{h}\left(.975\right)$} & {1.30} & {1.38} & {1.43} & {1.50} & {1.55} & {1.59} & {1.62} & {1.64} & {1.66} & {1.68} & {1.76} & {1.79} & {1.83}\tabularnewline
\bottomrule
\end{tabular}

\end{table}

\section{ICR Median-Unbiased Interval Estimator}\label{sec:MUE}

By definition, an estimator $\widehat{\theta}_{n}$ of a parameter
$\theta$ is median unbiased if $\P(\widehat{\theta}_{n}\geq\theta)\geq1/2$
and $\P(\widehat{\theta}_{n}\leq\theta)\geq1/2$. In this section,
we introduce an ICR MUE of $\rho$ that satisfies an analogous condition.
Also, with probability close to one, this interval estimator is a
point.\footnote{If the median of the asymptotic null distribution of the $t$-statistic
is strictly decreasing in $\rho$, or equivalently, if $c_{h}(0.5)$
is strictly increasing in $h$, then the proposed interval estimator
would be a point estimator with probability one. Because this condition
fails to hold exactly, but almost holds, there is a very small probability
that the estimator is a short interval rather than a point.}

Let $CI_{\text{ICR},n}^{\text{up}}(.5)$ and $CI_{\text{ICR},n}^{\text{low}}(.5)$
denote level $.5$ one-sided upper-bound and lower-bound CIs for $\rho$,
respectively. By definition,
\begin{align}
CI_{\text{ICR},n}^{\text{up}}(.5) & \coloneqq\left\{ \rho\in[-1+\epsilon,1]:c_{h}(0.5)\leq T_{n}(\rho)\ \text{for}\ h=n(1-\rho)\right\} \text{ and}\nonumber \\
CI_{\text{ICR},n}^{\text{low}}(.5) & \coloneqq\left\{ \rho\in[-1+\epsilon,1]:T_{n}(\rho)\leq c_{h}(0.5)\ \text{for}\ h=n(1-\rho)\right\} .\label{eq:CI_low}
\end{align}

The MUE $\widetilde{\rho}_{n}$ of $\rho$ is defined by 
\begin{align}
\widetilde{\rho}_{n} & =[\widetilde{\rho}_{n}^{\text{low}},\widetilde{\rho}_{n}^{\text{up}}],\text{ where}\nonumber \\
\widetilde{\rho}_{n}^{\text{up}} & =\max\{\rho:\rho\in CI_{\text{ICR},n}^{\text{up}}(.5)\}\text{ and}\nonumber \\
\widetilde{\rho}_{n}^{\text{low}} & =\min\{\rho:\rho\in CI_{\text{ICR},n}^{\text{low}}(.5)\}.\label{eq:rho_low}
\end{align}
By construction, we have $\widetilde{\rho}_{n}^{\text{low}}\leq\widetilde{\rho}_{n}^{\text{up}}$.\footnote{This holds because $\widetilde{\rho}_{n}^{\text{up}}\geq\sup\{\rho\in[-1+\epsilon,1]:c_{h}(.5)=T_{n}(\rho)\}$
and $\widetilde{\rho}_{n}^{\text{low}}$ is less than or equal to
the infimum of the values in the same set.} In addition, $\widetilde{\rho}_{n}$ is a singleton whenever the
set $\{\rho\in[-1+\epsilon,1]:T_{n}(\rho)=c_{h}(.5)\}$ contains a
single point, in which case $\widetilde{\rho}_{n}$ equals this point.
Table \ref{tab:ICR-quantiles} provides the critical values $c_{h}(.5)$
for a wide range of $h$. Given these critical values, computation
of $\widetilde{\rho}_{n}$ is fast.

\section{Monte Carlo Simulations}\label{sec:Monte-Carlo-Simulations}

We compare the coverage probabilities (CPs) and average lengths (ALs)
of the ICR and AG14 CIs  and report the absolute median biases (ABMs)
of the ICR MUE using Monte Carlo simulations. In Section \xrefblack{SM-sec:Additional-Simulation-Results}
of the Supplemental Material, we compare the performance of the ICR
and \textcite[Mik07 hereafter]{mikusheva2007uniform} CIs. We focus
on the nominal 95\% equal-tailed two-sided CIs. For the AG14 CI, we
follow the paper's constructions of the test statistic and critical
values.  To calculate the MUE for $\rho$, we follow the procedure
described in Section \ref{sec:MUE} and use $\widetilde{\rho}_{n}$
as the MUE of $\rho$ when $\widetilde{\rho}_{n}$ is a singleton.
When $\widetilde{\rho}_{n}^{\text{up}}\neq\widetilde{\rho}_{n}^{\text{low}}$,
we take $\widetilde{\rho}_{n}^{\text{up}}$ as the MUE for $\rho$
following \textcite[Section 5.1]{andrews2025inference}.

We consider a wide range of $\rho$ values: 0, 0.5, 0.7, 0.9, and
0.99. The innovations are of the form $U_{i}=\sigma_{i}\varepsilon_{i}$,
where $\{\varepsilon_{i}:i=0,\pm1,\dots\}$ are identically and independently
distributed (i.i.d.) standard normal and $\sigma_{i}$ is the multiplicative
conditional heteroskedasticity. Let GARCH-$(ma,ar;\psi)$ denote a
GARCH(1, 1) process with MA, AR, and intercept parameters $(ma,ar;\psi)$,
and let ARCH-$(ar_{1},\dots,ar_{4};\psi)$ denote an ARCH(4) process
with AR parameters $(ar_{1},\dots,ar_{4})$ and intercept $\psi$.
Following AG14, we consider five specifications for the conditional
heteroskedasticity of the innovations: (a) i.i.d. $N(0,1)$ (``i.i.d.''),
(b) GARCH(1,1)-(.05,.9;.001) (``GARCH1''), (c) GARCH(1,1)-(.15,.8;.2)
(``GARCH2''), (d) GARCH(1,1)-(.25,.7;.2) (``GARCH3''), and (e) ARCH(4)-(.3,.2,.2,.2;.2)
(``ARCH''). For the initial conditions, we consider the following
four specifications: (a) Fixed: $Y_{0}^{*}=0$, (b) Stationary: $Y_{0}^{*}=\sum_{i=0}^{\infty}\rho^{i}U_{-i}$,
(c) Scaled $n$: $Y_{0}^{*}\sim\sqrt{n}\cdot\sum_{i=0}^{\infty}\rho^{i}U_{-i}$,
and (d) Explosive: $Y_{0}^{*}\sim n^{3/4}\cdot\sum_{i=0}^{\infty}\rho^{i}U_{-i}$.
We consider a sample size of $n=150$. All simulation results are
based on 30,000 simulation repetitions.

\medskip{}

Tables \ref{tab:CP-ag14} and \ref{tab:CP-ICR} report the CPs of
AG14 and ICR CIs, respectively. The CPs of the AG14 CI in Table \ref{tab:CP-ag14}
are close to the nominal level of 95\% when the initial condition
is fixed at zero or stationary, but fall below 95\% when the initial
condition is more variable, namely in the scaled $n$ and explosive
$Y_{0}^{*}$ cases. In particular, the AG14 CI is robust to non-i.i.d.\ 
innovations when the initial condition is fixed or stationary.  On
the other hand, for scaled $n$ and explosive $Y_{0}^{*}$ cases,
the AG14 CPs lie in the interval {[}24.1, 93.5{]} with roughly a quarter
of the CPs being less than 79.0. Table \ref{tab:CP-ICR} shows that
the ICR CI achieves CPs close to 95\% in all of the scenarios considered
and for any initial conditions. Specifically, all ICR CI CPs lie in
the interval {[}93.5, 95.0{]}.

Table \ref{tab:length-ratio} reports the ratios of the ALs of the
nominal 95\% ICR CI to the AG14 CI. We include only cases with fixed
or stationary initial conditions because the AG14 CI is not robust
to more variable initial conditions, whereas the ICR CI is. The ratios
are between 1.00 and 1.11. On average, the ICR CI is 3.5\% longer
than the AG14 CI over the fifty cases considered.  Thus, the price
the ICR CI pays in terms of AL to gain robustness against the distribution
of the initial condition is fairly small.

\begin{table}[H]
\centering
\caption{Coverage probabilities ($\times100$) of the nominal 95\% AG14 CI}\label{tab:CP-ag14}

\begin{tabular}{lrrrrrlrrrrr}
\toprule
\multicolumn{1}{r}{\small\textbf{Initial Conditions:}} & \multicolumn{5}{c}{\small\textbf{Fixed}} & & \multicolumn{5}{c}{\small\textbf{Stationary}} \\
\multicolumn{1}{r}{\small\textbf{$\boldsymbol{\rho}$:}} & .00 & .50 & .70 & .90 & .99 & & .00 & .50 & .70 & .90 & .99 \\
\midrule
\small\textbf{i.i.d.} & 94.6 & 94.6 & 94.8 & 94.7 & 95.1 & & 94.6 & 94.5 & 94.8 & 94.7 & 94.5 \\
\small\textbf{GARCH1} & 94.6 & 94.5 & 94.8 & 94.8 & 95.1 & & 94.6 & 94.6 & 94.9 & 94.8 & 94.4 \\
\small\textbf{GARCH2} & 94.4 & 94.6 & 94.5 & 94.9 & 95.2 & & 94.3 & 94.6 & 94.5 & 94.8 & 94.7 \\
\small\textbf{GARCH3} & 94.2 & 94.3 & 94.6 & 94.7 & 95.2 & & 94.1 & 94.3 & 94.6 & 94.6 & 94.6 \\
\small\textbf{ARCH4} & 93.7 & 93.9 & 94.2 & 94.6 & 95.6 & & 93.6 & 94.0 & 94.1 & 94.5 & 94.7 \\
\midrule
\small\textbf{Initial Conditions:} & \multicolumn{5}{c}{\small\textbf{Scaled $\boldsymbol{n}$}} & & \multicolumn{5}{c}{\small\textbf{Explosive}} \\
\midrule
\small\textbf{i.i.d.} & 88.4 & 90.5 & 91.8 & 92.7 & 83.4 & & 63.7 & 80.5 & 86.6 & 90.8 & 77.7 \\
\small\textbf{GARCH1} & 59.1 & 79.0 & 86.4 & 91.7 & 83.2 & & 24.1 & 69.9 & 82.3 & 90.0 & 77.3 \\
\small\textbf{GARCH2} & 90.2 & 91.8 & 92.4 & 93.0 & 83.9 & & 71.0 & 83.3 & 87.8 & 90.9 & 77.4 \\
\small\textbf{GARCH3} & 90.9 & 92.1 & 92.5 & 93.4 & 84.2 & & 74.4 & 84.9 & 88.4 & 91.4 & 77.7 \\
\small\textbf{ARCH4} & 91.1 & 92.3 & 92.8 & 93.5 & 84.4 & & 76.7 & 86.3 & 89.4 & 91.8 & 77.9 \\
\bottomrule
\end{tabular}
\end{table}

\begin{table}[H]
\centering
\caption{Coverage probabilities ($\times100$) of the nominal 95\% ICR CI}\label{tab:CP-ICR}

\begin{tabular}{lrrrrr}
\toprule
\multicolumn{1}{r}{\small\textbf{Initial Conditions:}} & \multicolumn{5}{c}{\small\textbf{Arbitrary}} \\
\multicolumn{1}{r}{\small\textbf{$\boldsymbol{\rho}$:}} & .00 & .50 & .70 & .90 & .99 \\
\midrule
\small\textbf{i.i.d.} & 94.4 & 94.5 & 94.7 & 94.7 & 94.3 \\
\small\textbf{GARCH1} & 94.4 & 94.6 & 94.9 & 95.0 & 94.3 \\
\small\textbf{GARCH2} & 94.1 & 94.6 & 94.4 & 94.9 & 94.2 \\
\small\textbf{GARCH3} & 93.9 & 94.2 & 94.5 & 94.7 & 94.1 \\
\small\textbf{ARCH4} & 93.5 & 93.8 & 93.9 & 94.5 & 94.3 \\
\bottomrule
\end{tabular}
\end{table}

\begin{table}[H]
\centering
\caption{Ratios of the average lengths of the nominal 95\% ICR to AG14 CIs}\label{tab:length-ratio}

\begin{tabular}{lrrrrrlrrrrr}
\toprule
\multicolumn{1}{r}{\small\textbf{Initial Conditions:}} & \multicolumn{5}{c}{\small\textbf{Fixed}} & & \multicolumn{5}{c}{\small\textbf{Stationary}} \\
\multicolumn{1}{r}{\small\textbf{$\boldsymbol{\rho}$:}} & .00 & .50 & .70 & .90 & .99 & & .00 & .50 & .70 & .90 & .99 \\
\midrule
\small\textbf{i.i.d.} & 1.00 & 1.00 & 1.02 & 1.04 & 1.08 & & 1.00 & 1.01 & 1.02 & 1.05 & 1.11 \\
\small\textbf{GARCH1} & 1.04 & 1.05 & 1.06 & 1.07 & 1.08 & & 1.00 & 1.01 & 1.02 & 1.05 & 1.11 \\
\small\textbf{GARCH2} & 1.00 & 1.01 & 1.03 & 1.02 & 1.07 & & 1.00 & 1.01 & 1.03 & 1.04 & 1.11 \\
\small\textbf{GARCH3} & 1.00 & 1.01 & 1.03 & 1.02 & 1.06 & & 1.00 & 1.02 & 1.04 & 1.03 & 1.10 \\
\small\textbf{ARCH4} & 1.00 & 1.02 & 1.03 & 1.01 & 1.06 & & 1.01 & 1.02 & 1.04 & 1.03 & 1.10 \\
\bottomrule
\end{tabular}
\end{table}

\begin{table}[H]
\centering
\caption{Average lengths of the nominal 95\% ICR CI}\label{tab:avg_length_ICR}

\begin{tabular}{lrrrrrlrrrrr}
\toprule
\multicolumn{1}{r}{\small\textbf{Initial Conditions:}} & \multicolumn{5}{c}{\small\textbf{Fixed}} & & \multicolumn{5}{c}{\small\textbf{Stationary}} \\
\multicolumn{1}{r}{\small\textbf{$\boldsymbol{\rho}$:}} & .00 & .50 & .70 & .90 & .99 & & .00 & .50 & .70 & .90 & .99 \\
\midrule
\small\textbf{i.i.d.} & .32 & .28 & .24 & .17 & .08 & & .32 & .28 & .24 & .17 & .07 \\
\small\textbf{GARCH1} & .33 & .30 & .25 & .17 & .08 & & .34 & .30 & .26 & .17 & .08 \\
\small\textbf{GARCH2} & .38 & .34 & .29 & .19 & .08 & & .38 & .34 & .29 & .19 & .08 \\
\small\textbf{GARCH3} & .43 & .38 & .33 & .20 & .09 & & .43 & .38 & .33 & .20 & .08 \\
\small\textbf{ARCH4} & .49 & .44 & .37 & .21 & .09 & & .49 & .44 & .37 & .21 & .09 \\
\midrule
\small\textbf{Initial Conditions:} & \multicolumn{5}{c}{\small\textbf{Scaled $\boldsymbol{n}$}} & & \multicolumn{5}{c}{\small\textbf{Explosive}} \\
\midrule
\small\textbf{i.i.d.} & .31 & .27 & .23 & .14 & .04 & & .28 & .23 & .18 & .09 & .02 \\
\small\textbf{GARCH1} & .28 & .24 & .19 & .12 & .04 & & .18 & .15 & .12 & .07 & .02 \\
\small\textbf{GARCH2} & .37 & .32 & .26 & .15 & .04 & & .32 & .27 & .21 & .10 & .02 \\
\small\textbf{GARCH3} & .41 & .36 & .29 & .16 & .04 & & .36 & .30 & .23 & .11 & .02 \\
\small\textbf{ARCH4} & .46 & .40 & .32 & .17 & .04 & & .39 & .32 & .25 & .12 & .02 \\
\bottomrule
\end{tabular}
\end{table}

Table \ref{tab:avg_length_ICR} presents the ALs of the ICR CI. These
are substantially decreasing in $\rho$, as expected, longer for ARCH
conditional heteroskedasticity than i.i.d., and increasing from GARCH1
to GARCH2 to GARCH3 conditional heteroskedasticity.

\begin{table}[H]
\centering
\caption{Absolute Median Biases of the ICR MUE}\label{tab:mue_absolute_median_bias_ICR}

\begin{tabular}{lrrrrrlrrrrr}
\toprule
\multicolumn{1}{r}{\small\textbf{Initial Conditions:}} & \multicolumn{5}{c}{\small\textbf{Fixed}} & & \multicolumn{5}{c}{\small\textbf{Stationary}} \\
\multicolumn{1}{r}{\small\textbf{$\boldsymbol{\rho}$:}} & .00 & .50 & .70 & .90 & .99 & & .00 & .50 & .70 & .90 & .99 \\
\midrule
\small\textbf{i.i.d.} & .012 & .011 & .006 & .005 & .020 & & .012 & .011 & .006 & .005 & .015 \\
\small\textbf{GARCH1} & .017 & .011 & .006 & .005 & .020 & & .017 & .011 & .006 & .005 & .015 \\
\small\textbf{GARCH2} & .017 & .011 & .006 & .000 & .020 & & .017 & .011 & .006 & .000 & .015 \\
\small\textbf{GARCH3} & .022 & .011 & .011 & .000 & .020 & & .022 & .011 & .011 & .000 & .015 \\
\small\textbf{ARCH4} & .022 & .016 & .011 & .000 & .020 & & .022 & .016 & .011 & .000 & .015 \\
\midrule
\small\textbf{Initial Conditions:} & \multicolumn{5}{c}{\small\textbf{Scaled $\boldsymbol{n}$}} & & \multicolumn{5}{c}{\small\textbf{Explosive}} \\
\midrule
\small\textbf{i.i.d.} & .012 & .011 & .006 & .005 & .010 & & .012 & .011 & .006 & .010 & .010 \\
\small\textbf{GARCH1} & .017 & .011 & .006 & .005 & .010 & & .017 & .011 & .011 & .010 & .010 \\
\small\textbf{GARCH2} & .017 & .011 & .011 & .005 & .010 & & .017 & .011 & .011 & .010 & .010 \\
\small\textbf{GARCH3} & .022 & .011 & .011 & .005 & .010 & & .022 & .011 & .011 & .010 & .010 \\
\small\textbf{ARCH4} & .022 & .016 & .011 & .005 & .010 & & .022 & .016 & .011 & .010 & .010 \\
\bottomrule
\end{tabular}
\end{table}

Table \ref{tab:mue_absolute_median_bias_ICR} reports the AMBs of
the ICR MUE. Across all cases, the AMBs range from 0.000 to 0.022,
indicating that their magnitudes are generally small.

In Section \xrefblack{SM-sec:Additional-Simulation-Results} of the Supplemental
Material, we compare the performance of the ICR and Mik07 CIs via
simulations. The results are similar to those reported above for the
ICR and AG14 CIs except that the Mik07 CI undercovers both under conditional
heteroskedasticity, and/or scaled $n$ or explosive initial conditions,
but by a lower magnitude in the latter case than the AG14 CI. Again,
we find that the ICR CI is robust to more variable initial conditions
while paying only a small price in terms of CI width in the situations
where the Mik07 CI has correct asymptotic coverage, i.e., when the
innovations are i.i.d.\  and the initial conditions are stationary
(or fixed).

\section{Asymptotic Results}\label{sec:Asymptotic-Results}

This section establishes the correct uniform asymptotic size and asymptotic
similarity of the ICR CI for $\rho$.

\subsection{Parameter Space, Correct Uniform Size, and Asymptotic \\Similarity}\label{subsec:Parameter-space}

Following AG14, the parameter space for $(\rho,F)$ is given by:

\medskip{}

\noindent$\Lambda=\Big\{\lambda=(\rho,F)$: (i) \label{enu:para-space1}$\rho\in[-1+\epsilon,1]$;
\begin{enumerate}[start=2, nosep, leftmargin=4.5em]
\item \label{enu:para-space2}$\{U_{i}:i=1,2,\dots\}$ are stationary and
strong-mixing under $F$ with $\E_{F}(U_{i}|\mathcal{G}_{i-1})=0$
a.s., $\E_{F}(U_{i}^{2}|\mathcal{G}_{i-1})=\sigma_{i}^{2}$ a.s.,
where $\mathcal{G}_{i}$ is some non-decreasing sequence of $\sigma$-fields
for which $U_{j}\in\mathcal{G}_{i}$ for all $j\leq i$ for $i=1,2,\dots$;
\item \label{enu:para-space3}The strong-mixing numbers $\{\alpha_{F}(m):m\geq1\}$
satisfy $\alpha_{F}(m)\leq Cm^{-3\zeta/(\zeta-3)}$, $\forall m\geq1$;
\item \label{enu:para-space4}$\sup_{i,s,t}\E_{F}|\prod_{a\in A}a|^{\zeta}\leq M$,
where $0<i,s,t<\infty$, $i\geq\max(s,t)$, and $A$ is any non-empty
subset of $\{U_{i-s},U_{i-t},U_{i+1}^{2},U_{1}^{2}\}$;
\item \label{enu:para-space5}$\E_{F}U_{1}^{2}=1$;
\end{enumerate}
\hspace{2.6em}for some constants $0<\epsilon<2$, $\zeta>3$, and
$C<\infty$$\Big\}$.

\medskip{}

We suppose $\E_{F}U_{1}^{2}=1$ for simplicity because the asymptotic
distribution of $T_{n}(\rho)$ does not depend on the unconditional
variance. Parts \ref{enu:para-space2} and \ref{enu:para-space4}
are slightly different from the corresponding conditions in AG14 because
the random variables $\{U_{-j},\ j\geq0\}$ do not enter into the
asymptotic analysis of the ICR CI. The key innovation of this paper
is that we allow arbitrary initial conditions $Y_{0}^{*}$, including
those corresponding to explosive processes.

The main theoretical result of this paper shows that $CI_{\text{ICR},n}$
has correct asymptotic size for the parameter space $\L$ and is asymptotically
similar. Let $\P_{\lambda}$ denote probability under $\l=\left(\rho,F\right)\in\L$.
\begin{thm}
\label{thm:ICR-CP}Let $\a\in(0,1)$. For the parameter space $\L$,
the nominal $1-\alpha$ ICR CI for the parameter $\rho$ satisfies
\[
\underset{n\to\infty}{\lim\inf}\inf_{\lambda\in\Lambda}\P_{\lambda}(\rho\in CI_{\text{ICR},n})=\underset{n\to\infty}{\lim\sup}\sup_{\lambda\in\Lambda}\P_{\lambda}(\rho\in CI_{\text{ICR},n})=1-\alpha.
\]
\end{thm}

The MUE $\widetilde{\rho}_{n}$ has the following median-unbiasedness
property. This result is a corollary to the one-sided versions of
Theorem \ref{thm:ICR-CP}.\footnote{Corollary \ref{cor:MUE} holds because $CI_{\text{ICR},n}^{\text{up}}(.5)$
and $CI_{\text{ICR},n}^{\text{low}}(.5)$ both have coverage probabilities
of $1/2$ or greater by the proof of Theorem \ref{thm:ICR-CP} applied
to these one-sided CIs and $(\widetilde{\rho}_{n}\geq\rho)\supset CI_{\text{ICR},n}^{\text{up}}(.5)$
and $(\widetilde{\rho}_{n}\leq\rho)\supset CI_{\text{ICR},n}^{\text{low}}(.5)$.}
\begin{cor}
\label{cor:MUE}The MUE estimator $\widetilde{\rho}_{n}$ satisfies
\begin{align*}
\underset{n\to\infty}{\lim\inf}\inf_{\lambda\in\Lambda}\P_{\lambda}(\widetilde{\rho}_{n}^{\text{up}}\geq\rho) & \geq1/2\text{ and }\underset{n\to\infty}{\lim\inf}\inf_{\lambda\in\Lambda}\P_{\lambda}(\widetilde{\rho}_{n}^{\text{low}}\leq\rho)\geq1/2.
\end{align*}
\end{cor}

\subsection{Asymptotic Distribution}\label{subsec:Asymptotic-Distribution-Jh}

As stated in (\ref{eq:asym-distribution}), $J_{h}$ is the asymptotic
distribution of $T_{n}(\rho_{n})$ under suitable sequences $\{\rho_{n}:n\geq1\}$.
Here, we state this result formally as Theorem \ref{thm:Jh-convergence}
and define $J_{h}$. This result is proved in a similar way to Theorem
1 in \textcite{andrews2012asymptotics}. Using Theorem 2.1 in \textcite{andrews2020generic},
Theorem \ref{thm:Jh-convergence} is sufficient to establish Theorem
\ref{thm:ICR-CP}, which is proved in Section \xrefblack{SM-sec:Proof-of-Theorem1}
of the Supplemental Material.
\begin{thm}
\label{thm:Jh-convergence}For a sequence $\lambda_{n}=(\rho_{n},F_{n})$
such that $n(1-\rho_{n})\to h\in[0,\infty]$, we have
\[
T_{n}(\rho_{n})=\frac{n^{1/2}(\widehat{\rho}_{n}(\rho_{n})-\rho_{n})}{\widehat{\sigma}_{n}(\rho_{n})}\dto J_{h},
\]
where $J_{h}$ is defined immediately below.
\end{thm}
\begin{rem}
For any subsequence $\{p_{n}\}_{n\geq1}$ of $\{n\}_{n\geq1}$, Theorem
\ref{thm:Jh-convergence} holds with $p_{n}$ in place of $n$ throughout
and $h_{p_{n}}$ in place of $h=h_{n}$.
\end{rem}
For $h=\infty$, $J_{h}$ is the $N(0,1)$ distribution. For $h\in[0,\infty)$,
\begin{equation}
J_{h}:=\int_{0}^{1}I_{f,h}(r)dW(r)\Big/\left(\int_{0}^{1}I_{f,h}(r)^{2}dr\right)^{1/2},\label{eq:J_hintegral}
\end{equation}
where $I_{f,h}(r)$ is defined as follows.

Let $W(\cdot)$ denote a standard Brownian motion on $[0,1]$. For
$h\in(0,\infty)$, define $f_{h}(r)=(1,\exp(-hr))'$, and for $h=0$,
let $f_{h}(r)=(1,r)'$. Define 
\begin{align}
I_{h}(r) & :=\begin{cases}
\int_{0}^{r}\exp(-(r-s)h)dW(s) & \text{for }h>0\\
W(r) & \text{for }h=0
\end{cases}\quad\text{ and}\nonumber \\
I_{f,h}(r) & :=I_{h}(r)-\left(\int_{0}^{1}I_{h}(s)f_{h}(s)ds\right)'\left(\int_{0}^{1}f_{h}(s)f_{h}(s)'ds\right)^{-1}f_{h}(r),\label{eq:I_fh}
\end{align}
where the subscript $f$ indicates that we take the residual process
of $I_{h}(\cdot)$ after projecting it onto the space spanned by $f_{h}(\cdot)$.
When $h>0$, 
\begin{align}
I_{f,h}(r)=\  & I_{h}(r)-\alpha_{h}(r)\int_{0}^{1}I_{h}(s)ds-\beta_{h}(r)\int_{0}^{1}\frac{1-e^{-hs}}{h}I_{h}(s)ds,
\end{align}
where 
\begin{align}
\alpha_{h}(r) & =\ \frac{h[1-e^{-2h}-2(1-e^{-h})e^{-hr}+2he^{-hr}-2(1-e^{-h})]}{h(1-e^{-2h})-2(1-e^{-h})^{2}}\quad\text{and}\nonumber \\
\beta_{h}(r) & =\ \frac{2h^{2}[he^{-hr}-(1-e^{-h})]}{h(1-e^{-2h})-2(1-e^{-h})^{2}}.
\end{align}
When $h=0$, $I_{f,h}(r)$ reduces to detrended Brownian motion: 
\begin{equation}
W_{f}(r):=W(r)-(4-6r)\int_{0}^{1}W(s)ds-(12r-6)\int_{0}^{1}sW(s)ds.
\end{equation}
The following lemma establishes the connection between $I_{f,h}(r)$
for $h>0$ and $h=0$.
\begin{lem}
\label{lem:continuous-Ifh} Let $h\down0$, we have
\[
\sup_{r\in[0,1]}|\alpha_{h}(r)-(4-6r)|\to0\quad\text{and}\quad\sup_{r\in[0,1]}|\beta_{h}(r)-(12r-6)|\to0.
\]
\end{lem}
\begin{rem}
This lemma shows that $I_{f,h}(r)$ is uniformly continuous with respect
to $h$. Consequently, $I_{f,h}(r)\to I_{f,0}(r)$ uniformly over
$r\in[0,1]$ as $h\to0$ a.s.\  and the asymptotic distribution is
continuous at $h=0$.
\end{rem}
\newpage{}

\printbibliography[title={\centering\makebox[\textwidth]{References}}]

\clearpage{}

\newpage %\renewcommand{\thetable}{SM.\arabic{table}}

\makeatletter
\renewcommand{\thetable}{SM.\arabic{table}}
\renewcommand{\theHtable}{SM.\arabic{table}} % <-- make the hyperlink anchor unique, too
\setcounter{table}{0}
\makeatother

\clearpage
\thispagestyle{empty}

% use symbolic footnotes on this page only
\setcounter{footnote}{0}
\renewcommand{\thefootnote}{\fnsymbol{footnote}}

\begin{center}
\vspace*{2cm}

{\LARGE Supplemental Material}\\[1.5em]

{\LARGE to}\\[1.5em]

{\LARGE Initial-Condition-Robust Inference}\\[0.5em]
{\LARGE in Autoregressive Models}\\[3em]

\large
Donald W.\ K.\ Andrews\footnote{%
Andrews: Department of Economics and Cowles Foundation, Yale University.
\href{mailto:donald.andrews@yale.edu}{donald.andrews@yale.edu}.}
\qquad
Ming Li\footnote{%
Li: Department of Economics and Risk Management Institute,
National University of Singapore.
\href{mailto:mli@nus.edu.sg}{mli@nus.edu.sg}.}
\qquad
Yapeng Zheng\footnote{%
Zheng: Department of Economics,
Chinese University of Hong Kong.
\href{mailto:yapengzheng@link.cuhk.edu.hk}{yapengzheng@link.cuhk.edu.hk}.}\\[1.5em]

\large
\today

\end{center}

\clearpage
\setcounter{page}{1}

% restore normal footnote numbering
\renewcommand{\thefootnote}{\arabic{footnote}}

\newpage{}

\appendix
\setcounter{page}{1}

\begin{refsection}

In this Supplemental Material, Section \ref{sec:Computation-critical-values}
describes how the critical values are computed, Section \ref{sec:Proof-of-Theorem1}
proves Theorem \ref{thm:ICR-CP}, Section \ref{sec:Proof-of-Theorem2}
proves Theorem \ref{thm:Jh-convergence}, and Section \ref{sec:Additional-Simulation-Results}
provides some additional simulation results.

\section{Computation of Critical Values}\label{sec:Computation-critical-values}

For a given $\alpha$, we compute $c_{h}(\alpha)$ by simulating the
asymptotic distribution of $J_{h}$. To do so, we simulate $B=300,000$
sample paths of a standard Brownian motion on the interval $[0,1]$
by discretizing the interval into $N=50,000$ equal subintervals with
time step $1/N$. We initialize the process at $W(0)=0$ and construct
each path iteratively as 
\begin{equation}
W(r_{i})=W(r_{i-1})+Z_{i-1}/N^{1/2}\ \text{ for }r_{i}=i/N\text{ and }i=1,...,N,
\end{equation}
where $Z_{i-1}$ are independent standard normal random variables.
Next, we compute $I_{h}(r_{j})$ for $j=1,\dots,N$, using (\ref{eq:I_fh})
via a numerical approximation: 
\begin{equation}
I_{h}(r_{j})=\begin{cases}
\sum_{i=1}^{j}\exp(-(r_{j}-r_{i})h)\cdot(W(r_{i})-W(r_{i-1})) & \text{for }h>0,\\
W(r_{j}) & \text{for }h=0.
\end{cases}
\end{equation}
Then, we use (\ref{eq:I_fh}) to compute $I_{f,h}(r_{j})$ for $j=1,\dots,N$.
Given $I_{f,h}(r_{j})$ and $W(r_{j})$, we evaluate $J_{h}$ in (\ref{eq:J_hintegral})
via numerical approximation. Finally, the simulated estimate of $c_{h}(\alpha)$
is obtained as the $\alpha$-th quantile of the empirical distribution
of the $B$ realizations of $J_{h}$.

\section{Proof of Theorem \ref{thm:ICR-CP}}\label{sec:Proof-of-Theorem1}\label{SM-sec:Proof-of-Theorem1}

To prove Theorem \ref{thm:ICR-CP}, it is sufficient to verify Assumptions
B1{*} and B2 of \textcite{andrews2020generic} for $CS_{n}=CI_{\text{ICR},n}$.
In the present case $\lambda=(\rho,F)$, $r(\lambda)=\rho$, $h_{n}(\lambda)=n(1-\rho)\in R$,
$H=[0,\infty]$, and the parameter space $\Lambda$ is defined in
Section \ref{subsec:Parameter-space}.

For Assumption B1{*}, consider a sequence $\{\lambda_{n}=(\rho_{n},F_{n})\in\Lambda:n\geq1\}$
for which $h_{n}(\lambda_{n})\to h\in H$, i.e., $\rho_{n}=1-h_{n}/n$
and $h_{n}\to h\in[0,\infty]$. We have 
\begin{equation}
CP_{n}(\lambda_{n})=P_{\lambda_{n}}(\rho_{n}\in CI_{\text{ICR},n})=P_{\lambda_{n}}(c_{h_{n}}(\alpha/2)\le T_{n}(\rho_{n})\leq c_{h_{n}}(1-\alpha/2)).
\end{equation}
By Theorem \ref{thm:Jh-convergence}, we have $T_{n}(\rho_{n})\to_{d}J_{h}$
under $\{\lambda_{n}\in\Lambda_{n}:n\geq1\}$. In addition, $c_{h_{n}}(\beta)\to c_{h}(\beta)$
for $\beta=\alpha/2$ and $1-\alpha/2$. This is true because the
distribution function of $J_{h}$ is strictly increasing at its $\beta$-th
quantile. Thus, for any $\varepsilon>0$, we have
\begin{align}
L_{n}(c_{h}(\beta)-\varepsilon) & \to F_{J_{h}}(c_{h}(\beta)-\varepsilon)<\beta\text{ and }L_{n}(c_{h}(\beta)+\varepsilon)\to F_{J_{h}}(c_{h}(\beta)+\varepsilon)>\beta,
\end{align}
where $L_{n}(x)$ denotes the distribution function of $T_{n}(\rho_{n})$
at $x$ and $F_{J_{h}}(x)$ denotes the distribution function of $J_{h}$
at $x$. This and the definition $c_{h_{n}}(\beta)=\inf\{x\in R:F_{J_{h_{n}}}(x)\geq\beta\}$
yield $1\{c_{h}(\beta)-\varepsilon\leq c_{h_{n}}(\beta)\leq c_{h}(\beta)+\varepsilon\}\to1$
as $n\to\infty$ for any $\varepsilon>0$. By the definition of convergence
in distribution and continuity of $J_{h}$, it follows that $CP_{n}(\lambda_{n})\to1-\alpha$.
Assumption B1{*} therefore holds with $CP=1-\alpha$ for all $h\in H$.

For Assumption B2, assume we are given $\{\lambda_{p_{n}}\in\Lambda:n\geq1\}$
for a subsequence $\{p_{n}\}_{n\geq1}$ of $\{n\}_{n\geq1}$ such
that $h_{p_{n}}(\lambda_{p_{n}})\to h\in H$. Define $\{\lambda_{n}^{*}:n\geq1\}$
by (i) $\lambda_{p_{n}}^{*}=\lambda_{p_{n}}$ $\forall n\geq1$, (ii)
when $h<\infty$ and $m\neq p_{n}$, define $\lambda_{m}^{*}=(1-h/m,F^{*})$,
and (iii) when $h=\infty$ and $m\neq p_{n}$, define $\lambda_{m}^{*}=(0,F^{*})$,
where $F^{*}$ is the distribution such that $\{U_{i}:i=1,2,\dots\}$
are i.i.d. standard normal. Then, $\lambda_{n}^{*}\in\Lambda$ for
all $n\geq1$ and by construction $h_{n}(\lambda_{n}^{*})\to h\in H$.
This verifies Assumption B2 and completes the proof.

\section{Proof of Theorem \ref{thm:Jh-convergence} }\label{sec:Proof-of-Theorem2}

\subsection{Technical Lemmas}

Let 
\begin{align}
\widetilde{U} & =(U_{1},U_{2}+\rho_{n}U_{1},\dots,\sum_{j=1}^{n}\rho_{n}^{n-j}U_{j})',\nonumber \\
\widetilde{U}_{-1} & =(0,U_{1},\dots,\sum_{j=1}^{n-1}\rho_{n}^{n-1-j}U_{j})',\text{ and }\widetilde{U}_{0}=\sum_{j=0}^{\infty}\rho_{n}^{j}U_{-j}.
\end{align}
When $\rho_{n}\neq1$, define $S=(1-\rho_{n}^{n})/(1-\rho_{n})$ and
$Q=(1-\rho_{n}^{2n})/(1-\rho_{n}^{2})$. Let $[x]$ denote the integer
part of $x$.

The next lemma provides several useful results for the regime $h_{n}\coloneqq n(1-\rho_{n})\to h\in(0,\infty)$.
Unless stated otherwise, we let $n\gto\infty$ in presenting the asymptotic
results.
\begin{lem}
\label{lem:algebraic-relationship1}When $n(1-\rho_{n})\to h\in(0,\infty)$,
for a sequence $\{\l_{n}=(\rho_{n},F_{n})\in\L_{n}\}_{n\geq1}$, the
following results hold jointly\textup{:}
\begin{enumerate}[label=(\alph*), font=\upshape]
\item \label{enu:newnQ}$Q=\frac{1-e^{-2h}}{2h}n+o(n)$,
\item \label{enu:newnS}$S=\frac{1-e^{-h_{n}}}{h_{n}}n+o(n)$,
\item \label{enu:nQS1}$\frac{n^{2}}{nQ-S^{2}}\to\frac{2h^{2}}{h(1-e^{-2h})-2(1-e^{-h})^{2}}$
,
\item \label{enu:nQS2}$\frac{nS}{nQ-S^{2}}\to\frac{2h(1-e^{-h})}{h(1-e^{-2h})-2(1-e^{-h})^{2}}$,
\item \label{enu:nQS3}$\frac{nQ}{nQ-S^{2}}\to\frac{h(1-e^{-2h})}{h(1-e^{-2h})-2(1-e^{-h})^{2}}$,
\item \label{enu:nQS4}$n^{-1}\sum_{i=1}^{n}\left(\frac{n(Q-S\rho^{i-1})}{nQ-S^{2}}\right)^{2}=\frac{nQ}{nQ-S^{2}}$,
\item \label{enu:nQS5}$n^{-1}\sum_{i=1}^{n}\frac{n^{2}(Q-S\rho^{i-1})(n\rho^{i-1}-S)}{(nQ-S^{2})^{2}}=\frac{-nS}{(nQ-S^{2})}$,
and
\item \label{enu:nQS6}$n^{-1}\sum_{i=1}^{n}\left(\frac{n(n\rho^{i-1}-S)}{nQ-S^{2}}\right)^{2}=\frac{n^{2}}{nQ-S^{2}}$.
\end{enumerate}
\end{lem}
The next lemma provides several useful results for the regime $n(1-\rho_{n})\to0$.
\begin{lem}
\label{lem:algebraic-relationsip2}When $n(1-\rho_{n})\to h=0$, for
a sequence $\{\l_{n}=(\rho_{n},F_{n})\in\L_{n}\}_{n\geq1}$, the following
results hold jointly\textup{:}
\begin{enumerate}[label=(\alph*), font=\upshape]
\item \label{enu:al-relation1}$\frac{1-\rho_{n}^{i-1}}{n(1-\rho_{n})}\to r$
if $i=[nr]$,
\item \label{enu:al-relation2}$\frac{n(Q-2S+n)}{nQ-S^{2}}\to4$,
\item \label{enu:al-relation3}$\frac{n^{2}(1-\rho_{n})(n-S)}{nQ-S^{2}}\to6$,
and
\item \label{enu:al-relation4}$\frac{n^{4}(1-\rho_{n})^{2}}{nQ-S^{2}}\to12$.
\end{enumerate}
\end{lem}
The next lemma provides a set of convergence results that are used
in the proof of Theorem \ref{thm:Jh-convergence}.
\begin{lem}
\label{lem:FCLT}When $n(1-\rho_{n})\to h\in(0,\infty)$, for a sequence
$\{\l_{n}=(\rho_{n},F_{n})\in\L_{n}\}_{n\geq1}$, the following results
hold jointly\textup{:}
\begin{enumerate}[label=(\alph*), font=\upshape]
\item \label{enu:lem-a}$n^{-1/2}\widetilde{U}_{[nr]}\Rightarrow I_{h}(r)$,
\item \label{enu:lem-b}$n^{-1}\sum_{i=1}^{n}U_{i}\pto0$,
\item \label{enu:lem-c}$n^{-1}\sum_{i=1}^{n}U_{i}^{2}\pto1$,
\item \label{enu:lem-d}$n^{-1/2}\sum_{i=1}^{n}U_{i}\dto\int_{0}^{1}dW(r)$,
\item \label{enu:lem-e}$n^{-1/2}\sum_{i=1}^{n}\rho_{n}^{i-1}U_{i}\dto\int_{0}^{1}e^{-hr}dW(r)$,
\item \label{enu:lem-f}$n^{-3/2}\sum_{i=2}^{n}\widetilde{U}_{i-1}\dto\int_{0}^{1}I_{h}(r)dr$,
\item \label{enu:lem-g}$n^{-3/2}\sum_{i=2}^{n}\rho_{n}^{i-1}\widetilde{U}_{i-1}\dto\int_{0}^{1}e^{-hr}I_{h}(r)dr$,
\item \label{enu:lem-h}$n^{-1}\sum_{i=2}^{n}\widetilde{U}_{i-1}U_{i}\dto\int_{0}^{1}I_{h}(r)dW(r)$,
\item \label{enu:lem-i}$n^{-1}\sum_{i=2}^{n}\rho_{n}^{i-1}\widetilde{U}_{i-1}U_{i}\dto\int_{0}^{1}e^{-hr}I_{h}(r)dW(r)$,
\item \label{enu:lem-j}$n^{-2}\sum_{i=2}^{n}\widetilde{U}_{i-1}^{2}\dto\int_{0}^{1}I_{h}^{2}(r)dr$,
\item \label{enu:lem-k}$n^{-3/2}\sum_{i=2}^{n}\widetilde{U}_{i-1}U_{i}^{2}\dto\int_{0}^{1}I_{h}(r)dr$,
\item \label{enu:lem-l}$n^{-3/2}\sum_{i=2}^{n}\rho_{n}^{i-1}\widetilde{U}_{i-1}U_{i}^{2}\dto\int_{0}^{1}e^{-hr}I_{h}(r)dr$,
\item \label{enu:lem-m}$n^{-2}\sum_{i=2}^{n}\widetilde{U}_{i-1}^{2}U_{i}^{2}\dto\int_{0}^{1}I_{h}^{2}(r)dr$,
and
\item \label{enu:lem-n}$n^{-1-l_{1}/2}\sum_{i=2}^{n}\left(\rho_{n}^{i-1}\right)^{l_{0}}\widetilde{U}_{i-1}^{l_{1}}U_{i}^{l_{2}}=o_{p}(n)$
for $(l_{0},l_{1},l_{2})=(0,1,0)$, $(0,1,1)$, $(0,2,0)$, $(0,2,1)$,
$(0,3,0)$, $(0,3,1)$, $(0,4,0)$, $(1,1,1)$, $(1,2,0)$, $(1,3,0)$,
$(2,2,0)$, and $(1,3,0)$.
\end{enumerate}
\end{lem}
The following three lemmas are needed for the proof of the case $h=\infty$
in Theorem \ref{thm:Jh-convergence}. Since they are simplified versions
of Lemmas 6, 7, and 9 in \textcite[AG12 hereafter]{andrews2012asymptotics},
respectively, we omit their proofs.
\begin{lem}
\label{lem:Order-expect}When $n(1-\rho_{n})\to h=\infty$, for a
sequence $\{\l_{n}=(\rho_{n},F_{n})\in\L_{n}\}_{n\geq1}$, the following
results hold jointly\textup{:}
\begin{align*}
\E(\widetilde{U}_{0}^{2}U_{1}^{2})-(1-\rho_{n}^{2})^{-1}(\E U_{1}^{2})^{2} & =O(1),\\
\E\widetilde{U}_{0}^{2}-(1-\rho_{n}^{2})^{-1}\E U_{1}^{2} & =O(1),\\
\E\widetilde{U}_{0} & =O(1),\quad\text{and}\\
\E(\widetilde{U}_{0}U_{1}^{2}) & =O(1).
\end{align*}
\end{lem}
\begin{lem}
\label{lem:moment-condition}When $n(1-\rho_{n})\to h=\infty$, for
a sequence $\{\l_{n}=(\rho_{n},F_{n})\in\L_{n}\}_{n\geq1}$, we have
\[
\E\left(\sum_{i=2}^{n}\left[\E\zeta_{i}^{2}-\E(\zeta_{i}^{2}|\mathcal{G}_{i-1})\right]\right)^{2}\to0,\quad\text{where }\zeta_{i}\equiv n^{-1/2}\frac{\widetilde{U}_{i-1}U_{i}}{(\E(\widetilde{U}_{0}^{2}U_{1}^{2}))^{1/2}}.
\]
\end{lem}
\begin{lem}
\label{lem:integrable-condi}When $n(1-\rho_{n})\to h=\infty$, for
a sequence $\{\l_{n}=(\rho_{n},F_{n})\in\L_{n}\}_{n\geq1}$, we have
\[
\sum_{i=2}^{n}\E(\zeta_{i}^{2}1(|\zeta_{i}|>\delta)|\mathcal{G}_{i-1})\pto0
\]
for any $\delta>0$.
\end{lem}
The next lemma is slightly different from Lemma 8 in AG12, so we prove
the new results. In this lemma, we redefine $X_{2}(\rho_{n})$ as
an $n\times2$ matrix whose $i$-th row is $(1,n^{1/2}(1-\rho_{n})^{1/2}\rho_{n}^{i-1})$
and let $X(\rho_{n})=[X_{1}:X_{2}(\rho_{n})]$. Note that this redefinition
does not affect $P_{X_{2}(\rho_{n})}$ or $M_{X_{2}(\rho_{n})}$.
\begin{lem}
\label{lem:case2lem}When $n(1-\rho_{n})\to h=\infty$, for a sequence
$\{\l_{n}=(\rho_{n},F_{n})\in\L_{n}\}_{n\geq1}$, the following results
hold jointly\textup{:}
\begin{enumerate}[label=(\alph*), font=\upshape]
\item \label{enu:case2lem-a}$n^{-1}(1-\rho_{n})^{1/2}\widetilde{U}_{-1}'X_{2}(\rho_{n})=(o_{p}(1),o_{p}(1))$,
\item \label{enu:case2lem-b}$(\E\widetilde{U}_{0}^{2})^{-1}n^{-1}\widetilde{U}_{-1}'\widetilde{U}_{-1}\pto1$,
\item \label{enu:case2lem-c}$(\E(\widetilde{U}_{0}^{2}U_{1}^{2}))^{-1}n^{-1}\sum_{i=1}^{n}\widetilde{U}_{i-1}^{2}U_{i}^{2}\pto1$,
\item \label{enu:case2lem-d}$(X(\rho_{n})'X(\rho_{n}))^{-1}X(\rho_{n})'U=(O_{p}(n^{-1/2}(1-\rho_{n})^{3/2}),O_{p}(n^{-1/2}(1-\rho_{n})),O_{p}(n^{-1/2}(1-\rho_{n})))'$,
\item \label{enu:case2lem-e}$(\E(\widetilde{U}_{0}^{2}U_{1}^{2}))^{-1}n^{-1}\widetilde{U}_{-1}'\Delta^{2}\widetilde{U}_{-1}\pto1$,
\item \label{enu:case2lem-f}$(1-\rho_{n})^{1/2}n^{-1}(X_{2}(\rho_{n})'\Delta^{2}\widetilde{U}_{-1})=O_{p}(1)$,
and
\item \label{enu:case2lem-g}$n^{-1}(X_{2}(\rho_{n})'\Delta^{2}X_{2}(\rho_{n}))=O_{p}(1)$.
\end{enumerate}
\end{lem}

\subsection{Proof of Theorem \ref{thm:Jh-convergence}}

\subsubsection{Local-to-Unity Case}

We first handle the case for $h\in(0,\infty)$. We can write 
\begin{align}
n(\widehat{\rho}_{n}(\rho_{n})-\rho_{n})=\  & (n^{-2}X_{1}'M_{X_{2}(\rho_{n})}X_{1})^{-1}n^{-1}X_{1}'M_{X_{2}(\rho_{n})}U\quad\text{and}\nonumber \\
n\widehat{\sigma}_{n}^{2}(\rho_{n})=\  & (n^{-2}X_{1}'M_{X_{2}(\rho_{n})}X_{1})^{-1}(n^{-2}X_{1}'M_{X_{2}(\rho_{n})}\Delta^{2}M_{X_{2}(\rho_{n})}X_{1})\nonumber \\
 & \times(n^{-2}X_{1}'M_{X_{2}(\rho_{n})}X_{1})^{-1}.\label{eq:rho-sigma-expansion}
\end{align}
We consider the terms in (\ref{eq:rho-sigma-expansion}) one at a
time. First, we have 
\begin{equation}
M_{X_{2}(\rho_{n})}X_{1}=M_{X_{2}(\rho_{n})}(X_{2}(\rho_{n})(\mu,Y_{0}^{*})'+\widetilde{U}_{-1})=M_{X_{2}(\rho_{n})}\widetilde{U}_{-1}.
\end{equation}
In consequence,
\begin{align}
 & n^{-2}X_{1}'M_{X_{2}(\rho_{n})}X_{1}\nonumber \\
=\  & n^{-2}(M_{X_{2}(\rho_{n})}\widetilde{U}_{-1})'\widetilde{U}_{-1}\nonumber \\
=\  & n^{-2}\sum_{i=2}^{n}\widetilde{U}_{i-1}\left(\widetilde{U}_{i-1}-\frac{1}{nQ-S^{2}}\sum_{k=2}^{n}\left(Q-S(\rho_{n}^{k-1}+\rho_{n}^{i-1})+n\rho_{n}^{i+k-2}\right)\widetilde{U}_{k-1}\right)\nonumber \\
=\  & n^{-2}\sum_{i=2}^{n}\widetilde{U}_{i-1}^{2}-n^{-2}\frac{1}{nQ-S^{2}}\sum_{i=2}^{n}\widetilde{U}_{i-1}\sum_{k=2}^{n}\left(Q-S(\rho_{n}^{k-1}+\rho_{n}^{i-1})+n\rho_{n}^{i+k-2}\right)\widetilde{U}_{k-1}\nonumber \\
=\  & n^{-2}\sum_{i=2}^{n}\widetilde{U}_{i-1}^{2}-\frac{nQ}{nQ-S^{2}}\left(n^{-3/2}\sum_{i=2}^{n}\widetilde{U}_{i-1}\right)^{2}\nonumber \\
 & +\frac{2nS}{nQ-S^{2}}\left(n^{-3/2}\sum_{i=2}^{n}\widetilde{U}_{i-1}\right)\left(n^{-3/2}\sum_{i=2}^{n}\rho_{n}^{i-1}\widetilde{U}_{i-1}\right)\nonumber \\
 & -\frac{n^{2}}{nQ-S^{2}}\left(n^{-3/2}\sum_{i=2}^{n}\rho_{n}^{i-1}\widetilde{U}_{i-1}\right)^{2}\nonumber \\
\dto\  & \int_{0}^{1}I_{h}^{2}(r)dr-\frac{h(1-e^{-2h})}{h(1-e^{-2h})-2(1-e^{-h})^{2}}\left(\int_{0}^{1}I_{h}(r)dr\right)^{2}\nonumber \\
 & +\frac{4h(1-e^{-h})}{h(1-e^{-2h})-2(1-e^{-h})^{2}}\left(\int_{0}^{1}I_{h}(r)dr\right)\left(\int_{0}^{1}e^{-hr}I_{h}(r)dr\right)\nonumber \\
 & -\frac{2h^{2}}{h(1-e^{-2h})-2(1-e^{-h})^{2}}\left(\int_{0}^{1}e^{-hr}I_{h}(r)dr\right)^{2}\nonumber \\
=\  & \int_{0}^{1}I_{f,h}^{2}(r)dr,\label{eq:bread-asy}
\end{align}
where the convergence holds by Lemma \ref{lem:algebraic-relationship1}
and parts \ref{enu:lem-f} and \ref{enu:lem-g} of Lemma \ref{lem:FCLT}
and the last equality holds by (\ref{eq:I_fh}).

Similarly, we have 
\begin{align}
 & n^{-1}X_{1}'M_{X_{2}(\rho_{n})}U\nonumber \\
=\  & n^{-1}(M_{X_{2}(\rho_{n})}\widetilde{U}_{-1})'U\nonumber \\
=\  & n^{-1}\sum_{i=2}^{n}\left(\widetilde{U}_{i-1}-\frac{1}{nQ-S^{2}}\sum_{k=2}^{n}\left(Q-S(\rho_{n}^{k-1}+\rho_{n}^{i-1})+n\rho_{n}^{i+k-2}\right)\widetilde{U}_{k-1}\right)U_{i}\nonumber \\
=\  & n^{-1}\sum_{i=2}^{n}\widetilde{U}_{i-1}U_{i}-\frac{nQ}{nQ-S^{2}}\left(n^{-1/2}\sum_{i=1}^{n}U_{i}\right)\left(n^{-3/2}\sum_{i=2}^{n}\widetilde{U}_{i-1}\right)\nonumber \\
 & +\frac{nS}{nQ-S^{2}}\left(n^{-1/2}\sum_{i=1}^{n}U_{i}\right)\left(n^{-3/2}\sum_{i=2}^{n}\rho_{n}^{i-1}\widetilde{U}_{i-1}\right)\nonumber \\
 & +\frac{nS}{nQ-S^{2}}\left(n^{-1/2}\sum_{i=1}^{n}\rho_{n}^{i-1}U_{i}\right)\left(n^{-3/2}\sum_{i=2}^{n}\widetilde{U}_{i-1}\right)\nonumber \\
 & -\frac{n^{2}}{nQ-S^{2}}\left(n^{-1/2}\sum_{i=1}^{n}\rho_{n}^{i-1}U_{i}\right)\left(n^{-3/2}\sum_{i=2}^{n}\rho_{n}^{i-1}\widetilde{U}_{i-1}\right)\nonumber \\
\dto\  & \int_{0}^{1}I_{h}(r)dW(r)-\frac{h(1-e^{-2h})}{h(1-e^{-2h})-2(1-e^{-h})^{2}}\left(\int_{0}^{1}dW(r)\right)\left(\int_{0}^{1}I_{h}(r)dr\right)\nonumber \\
 & +\frac{2h(1-e^{-h})}{h(1-e^{-2h})-2(1-e^{-h})^{2}}\left(\int_{0}^{1}dW(r)\right)\left(\int_{0}^{1}e^{-hr}I_{h}(r)dr\right)\nonumber \\
 & +\frac{2h(1-e^{-h})}{h(1-e^{-2h})-2(1-e^{-h})^{2}}\left(\int_{0}^{1}e^{-hr}dW(r)\right)\left(\int_{0}^{1}I_{h}(r)dr\right)\nonumber \\
 & -\frac{2h^{2}}{h(1-e^{-2h})-2(1-e^{-h})^{2}}\left(\int_{0}^{1}e^{-hr}dW(r)\right)\left(\int_{0}^{1}e^{-hr}I_{h}(r)dr\right)\nonumber \\
=\  & \int_{0}^{1}I_{f,h}(r)dW(r),\label{eq:numerator}
\end{align}
where the convergence uses Lemma \ref{lem:algebraic-relationship1}
and parts \ref{enu:lem-d}--\ref{enu:lem-g} of Lemma \ref{lem:FCLT}.

To determine the asymptotic distribution of $n^{-2}X_{1}'M_{X_{2}(\rho_{n})}\Delta^{2}M_{X_{2}(\rho_{n})}X_{1}$,
we proceed as follows. The space spanned by $X(\rho_{n})$ is equivalent
to the space spanned by $\{(1,\rho_{n}^{i-1},n^{-1/2}\widetilde{U}_{i-1})\}_{i=1}^{n}$.
Therefore, we have
\begin{align}
\widehat{U}_{i} & =U_{i}-A_{n}'B_{n}^{-1}C_{n,i},\text{ where}\nonumber \\
A_{n} & =\begin{pmatrix}n^{-1/2}\sum_{j=1}^{n}U_{j}\\
n^{-1/2}\sum_{j=1}^{n}\rho_{n}^{j-1}U_{j}\\
n^{-1}\sum_{j=2}^{n}\widetilde{U}_{j-1}U_{j}
\end{pmatrix},\nonumber \\
B_{n} & =\begin{pmatrix}1 & n^{-1}S & n^{-3/2}\sum_{j=2}^{n}\widetilde{U}_{j-1}\\
n^{-1}S & n^{-1}Q & n^{-3/2}\sum_{j=2}^{n}\rho_{n}^{j-1}\widetilde{U}_{j-1}\\
n^{-3/2}\sum_{j=2}^{n}\widetilde{U}_{j-1} & n^{-3/2}\sum_{j=2}^{n}\rho_{n}^{j-1}\widetilde{U}_{j-1} & n^{-2}\sum_{j=2}^{n}\widetilde{U}_{j-1}^{2}
\end{pmatrix},\text{ and}\nonumber \\
C_{n,i} & =\begin{pmatrix}n^{-1/2}\\
n^{-1/2}\rho_{n}^{i-1}\\
n^{-1}\widetilde{U}_{i-1}
\end{pmatrix}.\label{eq:residual-Ui}
\end{align}

For each term in $A_{n}$, we have 
\begin{align}
n^{-1/2}\sum_{j=1}^{n}U_{j} & \dto\int_{0}^{1}dW(r),\nonumber \\
n^{-1/2}\sum_{j=1}^{n}\rho_{n}^{j-1}U_{j} & \dto\int_{0}^{1}e^{-hr}dW(r),\text{ and}\nonumber \\
n^{-1}\sum_{j=2}^{n}\widetilde{U}_{j-1}U_{j} & =\int_{0}^{1}I_{h}(r)dW(r)
\end{align}
by Lemma \ref{lem:FCLT}\ref{enu:lem-d}, \ref{enu:lem-e}, and \ref{enu:lem-h},
respectively. Therefore, $A_{n}=O_{p}(1)$.

For the entries of $B_{n}$, first we have 
\begin{align}
n^{-1}S & =\frac{1-e^{-h}}{h}+o(1)=O(1)\text{ and}\nonumber \\
n^{-1}Q & =\frac{1-e^{-2h}}{2h}+o(1)=O(1)
\end{align}
by Lemma \ref{lem:algebraic-relationship1}\ref{enu:newnQ} and \ref{enu:newnS}.
Using Lemma \ref{lem:FCLT}\ref{enu:lem-f}, \ref{enu:lem-g}, and
\ref{enu:lem-j}, we have 
\begin{align}
n^{-3/2}\sum_{j=2}^{n}\widetilde{U}_{j-1} & \dto\int_{0}^{1}I_{h}(r)dr,\nonumber \\
n^{-3/2}\sum_{j=2}^{n}\rho_{n}^{j-1}\widetilde{U}_{j-1} & \dto\int_{0}^{1}e^{-hr}dW(r),\text{ and}\nonumber \\
n^{-2}\sum_{j=2}^{n}\widetilde{U}_{j-1}^{2} & \dto\int_{0}^{1}I_{h}^{2}(r)dr.
\end{align}
Therefore, $B_{n}=O_{p}(1)$. Also, note that
\begin{align}
\frac{n^{2}}{nQ-S^{2}}\det(B_{n})=\  & n^{-2}\sum_{j=2}^{n}\widetilde{U}_{j-1}^{2}-\frac{n^{2}}{nQ-S^{2}}\left(n^{-3/2}\sum_{j=2}^{n}\rho_{n}^{j-1}\widetilde{U}_{j-1}\right)^{2}\nonumber \\
 & +2\frac{nS}{nQ-S^{2}}\left(n^{-3/2}\sum_{j=2}^{n}\widetilde{U}_{j-1}\right)\left(n^{-3/2}\sum_{j=2}^{n}\rho_{n}^{j-1}\widetilde{U}_{j-1}\right)\nonumber \\
 & -\frac{nQ}{nQ-S^{2}}\left(n^{-3/2}\sum_{j=2}^{n}\widetilde{U}_{j-1}\right)^{2}\nonumber \\
\dto\  & \int_{0}^{1}I_{h}^{2}(r)dr-\frac{2h^{2}}{h(1-e^{-2h})-2(1-e^{-h})^{2}}\left(\int_{0}^{1}e^{-hr}I_{h}(r)dr\right)^{2}\nonumber \\
 & +\frac{4h(1-e^{-h})}{h(1-e^{-2h})-2(1-e^{-h})^{2}}\left(\int_{0}^{1}I_{h}(r)dr\right)\left(\int_{0}^{1}e^{-hr}I_{h}(r)dr\right)\nonumber \\
 & -\frac{h(1-e^{-2h})}{h(1-e^{-2h})-2(1-e^{-h})^{2}}\left(\int_{0}^{1}I_{h}(r)dr\right)^{2}\nonumber \\
=\  & \int_{0}^{1}I_{f,h}^{2}(r)dr.
\end{align}
Since $\P(\int_{0}^{1}I_{f,h}^{2}(r)dr=0)=0$, we have $\P(\det(B_{n})=0)\to0$
as $n\to\infty$, which implies that $B_{n}$ is invertible with probability
approaching one. Combining this with the fact that $B_{n}=O_{p}(1)$,
we have $B_{n}^{-1}=O_{p}(1)$.

Next, we define
\begin{align}
D_{n} & =\begin{pmatrix}n^{-3/2}\sum_{i=2}^{n}\widetilde{U}_{i-1}^{2}U_{i}\\
n^{-3/2}\sum_{i=2}^{n}\widetilde{U}_{i-1}^{2}\rho_{n}^{i-1}U_{i}\\
n^{-2}\sum_{i=2}^{n}\widetilde{U}_{i-1}^{3}U_{i}
\end{pmatrix}\text{ and}\nonumber \\
K_{n} & =\begin{pmatrix}n^{-2}\sum_{i=2}^{n}\widetilde{U}_{i-1}^{2} & n^{-2}\sum_{i=2}^{n}\rho_{n}^{i-1}\widetilde{U}_{i-1}^{2} & n^{-5/2}\sum_{i=2}^{n}\widetilde{U}_{i-1}^{3}\\
n^{-2}\sum_{i=2}^{n}\rho_{n}^{i-1}\widetilde{U}_{i-1}^{2} & n^{-2}\sum_{i=2}^{n}\rho_{n}^{2i-2}\widetilde{U}_{i-1}^{2} & n^{-5/2}\sum_{i=2}^{n}\rho_{n}^{i-1}\widetilde{U}_{i-1}^{3}\\
n^{-5/2}\sum_{i=2}^{n}\widetilde{U}_{i-1}^{3} & n^{-5/2}\sum_{i=2}^{n}\rho_{n}^{i-1}\widetilde{U}_{i-1}^{3} & n^{-3}\sum_{i=2}^{n}\widetilde{U}_{i-1}^{4}
\end{pmatrix}.
\end{align}
Using the results of Lemma \ref{lem:FCLT}\ref{enu:lem-n} with $(l_{0},l_{1},l_{2})=(0,2,1)$,
$(1,1,1)$, and $(0,3,1)$, we have $D_{n}=o_{p}(n)$. Similarly,
using Lemma \ref{lem:FCLT}\ref{enu:lem-j} and Lemma \ref{lem:FCLT}\ref{enu:lem-n}
with $(l_{0},l_{1},l_{2})=(2,2,0)$, $(0,3,0)$, $(1,3,0)$, and $(0,4,0)$,
we have $K_{n}=o_{p}(n)$.

By (\ref{eq:residual-Ui}), we have 
\begin{align}
n^{-2}\sum_{i=2}^{n}\widetilde{U}_{i-1}^{2}\widehat{U}_{i}^{2} & =n^{-2}\sum_{i=2}^{n}\widetilde{U}_{i-1}^{2}U_{i}^{2}-2n^{-1}A_{n}'B_{n}^{-1}D_{n}+n^{-1}A_{n}'B_{n}^{-1}K_{n}B_{n}^{-1}A_{n}\nonumber \\
 & =n^{-2}\sum_{i=2}^{n}\widetilde{U}_{i-1}^{2}U_{i}^{2}+o_{p}(1),\label{eq:Ui1-squareUihat}
\end{align}
where the second equality holds by $A_{n}=O_{p}(1)$, $B_{n}^{-1}=O_{p}(1)$,
$D_{n}=o_{p}(n)$, and $K_{n}=o_{p}(n)$.

Proceeding analogously to (\ref{eq:Ui1-squareUihat}), replacing $\widetilde{U}_{i-1}^{2}$
with $\widetilde{U}_{i-1}$ and then omitting the $\widetilde{U}_{i-1}^{2}$
term, we have
\begin{align}
n^{-3/2}\sum_{i=2}^{n}\widetilde{U}_{i-1}\widehat{U}_{i}^{2} & =n^{-3/2}\sum_{i=2}^{n}\widetilde{U}_{i-1}U_{i}^{2}+o_{p}(1)\quad\text{and}\nonumber \\
n^{-1}\sum_{i=1}^{n}c_{i}\widehat{U}_{i}^{2} & =n^{-1}\sum_{i=1}^{n}c_{i}U_{i}^{2}+o_{p}(1)\label{eq:Ui-relationships}
\end{align}
for any non-random sequence $c_{i}=O(1)$.

Next, we have
\begin{align}
 & n^{-2}X_{1}'M_{X_{2}(\rho_{n})}\Delta^{2}M_{X_{2}(\rho_{n})}X_{1}\nonumber \\
=\  & n^{-2}\sum_{i=1}^{n}\frac{\widehat{U}_{i}^{2}}{(1-p_{ii}^{*})^{2}}\left(\widetilde{U}_{i-1}-\frac{1}{nQ-S^{2}}\sum_{k=2}^{n}\left(Q-S(\rho_{n}^{k-1}+\rho_{n}^{i-1})+n\rho_{n}^{i+k-2}\right)\widetilde{U}_{k-1}\right)^{2}\nonumber \\
=\  & n^{-2}\sum_{i=1}^{n}\widehat{U}_{i}^{2}\left(\widetilde{U}_{i-1}-\frac{1}{nQ-S^{2}}\sum_{k=2}^{n}\left(Q-S(\rho_{n}^{k-1}+\rho_{n}^{i-1})+n\rho_{n}^{i+k-2}\right)\widetilde{U}_{k-1}\right)^{2}+o_{p}(1)\nonumber \\
=\  & n^{-2}\sum_{i=2}^{n}\widetilde{U}_{i-1}^{2}\widehat{U}_{i}^{2}-2n^{-2}\frac{1}{nQ-S^{2}}\sum_{i=2}^{n}\sum_{k=2}^{n}\left(Q-S(\rho_{n}^{k-1}+\rho_{n}^{i-1})+n\rho_{n}^{i+k-2}\right)\widetilde{U}_{k-1}\widetilde{U}_{i-1}\widehat{U}_{i}^{2}\nonumber \\
 & +n^{-2}\frac{1}{(nQ-S^{2})^{2}}\sum_{i=1}^{n}\left(\sum_{k=2}^{n}\left(Q-S(\rho_{n}^{k-1}+\rho_{n}^{i-1})+n\rho_{n}^{i+k-2}\right)\widetilde{U}_{k-1}\right)^{2}\widehat{U}_{i}^{2}+o_{p}(1)\nonumber \\
=\  & \mathcal{U}_{1,n}-2\mathcal{U}_{2,n}+\mathcal{U}_{3n}+o_{p}(1),\label{eq:meat}
\end{align}
where the second equality holds by $p_{ii}^{*}=\min\{p_{ii},n^{-1/2}\}=O_{p}(n^{-1/2})$.
We now analyze $\mathcal{U}_{1,n}$, $\mathcal{U}_{2,n}$, and $\mathcal{U}_{3,n}$
in turn.

For $\mathcal{U}_{1,n}$, by (\ref{eq:Ui1-squareUihat}) and Lemma
\ref{lem:FCLT}\ref{enu:lem-m}, we have
\begin{equation}
\mathcal{U}_{1,n}=n^{-2}\sum_{i=2}^{n}\widetilde{U}_{i-1}^{2}U_{i}^{2}+o_{p}(1)\dto\int_{0}^{1}I_{h}^{2}(r)dr.\label{eq:U1n}
\end{equation}

For $\mathcal{U}_{2,n}$, we have 
\begin{align}
\mathcal{U}_{2,n}=\  & \frac{nQ}{nQ-S^{2}}\left(n^{-3/2}\sum_{i=2}^{n}\widetilde{U}_{i-1}\right)\left(n^{-3/2}\sum_{i=2}^{n}\widetilde{U}_{i-1}\widehat{U}_{i}^{2}\right)\nonumber \\
 & -\frac{nS}{nQ-S^{2}}\left(n^{-3/2}\sum_{i=2}^{n}\rho_{n}^{i-1}\widetilde{U}_{i-1}\right)\left(n^{-3/2}\sum_{i=2}^{n}\widetilde{U}_{i-1}\widehat{U}_{i}^{2}\right)\nonumber \\
 & -\frac{nS}{nQ-S^{2}}\left(n^{-3/2}\sum_{i=2}^{n}\widetilde{U}_{i-1}\right)\left(n^{-3/2}\sum_{i=2}^{n}\rho_{n}^{i-1}\widetilde{U}_{i-1}\widehat{U}_{i}^{2}\right)\nonumber \\
 & +\frac{n^{2}}{nQ-S^{2}}\left(n^{-3/2}\sum_{i=2}^{n}\rho_{n}^{i-1}\widetilde{U}_{i-1}\right)\left(n^{-3/2}\sum_{i=2}^{n}\rho_{n}^{i-1}\widetilde{U}_{i-1}\widehat{U}_{i}^{2}\right)\nonumber \\
=\  & \frac{nQ}{nQ-S^{2}}\left(n^{-3/2}\sum_{i=2}^{n}\widetilde{U}_{i-1}\right)\left(n^{-3/2}\sum_{i=2}^{n}\widetilde{U}_{i-1}U_{i}^{2}\right)\nonumber \\
 & -\frac{nS}{nQ-S^{2}}\left(n^{-3/2}\sum_{i=2}^{n}\rho_{n}^{i-1}\widetilde{U}_{i-1}\right)\left(n^{-3/2}\sum_{i=2}^{n}\widetilde{U}_{i-1}U_{i}^{2}\right)\nonumber \\
 & -\frac{nS}{nQ-S^{2}}\left(n^{-3/2}\sum_{i=2}^{n}\widetilde{U}_{i-1}\right)\left(n^{-3/2}\sum_{i=2}^{n}\rho_{n}^{i-1}\widetilde{U}_{i-1}U_{i}^{2}\right)\nonumber \\
 & +\frac{n^{2}}{nQ-S^{2}}\left(n^{-3/2}\sum_{i=2}^{n}\rho_{n}^{i-1}\widetilde{U}_{i-1}\right)\left(n^{-3/2}\sum_{i=2}^{n}\rho_{n}^{i-1}\widetilde{U}_{i-1}U_{i}^{2}\right)+o_{p}(1)\nonumber \\
\dto\  & \frac{h(1-e^{-2h})}{h(1-e^{-2h})-2(1-e^{-h})^{2}}\left(\int_{0}^{1}I_{h}(r)dr\right)^{2}\nonumber \\
 & -\frac{4h(1-e^{-h})}{h(1-e^{-2h})-2(1-e^{-h})^{2}}\left(\int_{0}^{1}e^{-hr}I_{h}(r)dr\right)\left(\int_{0}^{1}I_{h}(r)dr\right)\nonumber \\
 & +\frac{2h^{2}}{h(1-e^{-2h})-2(1-e^{-h})^{2}}\left(\int_{0}^{1}e^{-hr}I_{h}(r)dr\right)^{2},\label{eq:U2n}
\end{align}
where the first equality holds by the definition of $\mathcal{U}_{2,n}$
and algebra, the second equality holds by (\ref{eq:Ui-relationships}),
and the convergence holds by Lemma \ref{lem:algebraic-relationship1}\ref{enu:nQS1}--\ref{enu:nQS3}
and Lemma \ref{lem:FCLT}\ref{enu:lem-f}, \ref{enu:lem-g}, \ref{enu:lem-k},
and \ref{enu:lem-l}.

For $\mathcal{U}_{3,n}$, first notice that
\begin{align}
 & \left(\sum_{k=2}^{n}\left(Q-S(\rho_{n}^{k-1}+\rho_{n}^{i-1})+n\rho^{i+k-2}\right)\widetilde{U}_{k-1}\right)^{2}\nonumber \\
=\  & \left((Q-S\rho^{i-1})\left(\sum_{k=2}^{n}\widetilde{U}_{k-1}\right)+(n\rho^{i-1}-S)\left(\sum_{k=2}^{n}\rho_{n}^{k-1}\widetilde{U}_{k-1}\right)\right)^{2}\nonumber \\
=\  & (Q-S\rho^{i-1})^{2}\left(\sum_{k=2}^{n}\widetilde{U}_{k-1}\right)^{2}+2(Q-S\rho^{i-1})(n\rho^{i-1}-S)\left(\sum_{k=2}^{n}\widetilde{U}_{k-1}\right)\left(\sum_{k=2}^{n}\rho_{n}^{k-1}\widetilde{U}_{k-1}\right)\nonumber \\
 & +(n\rho^{i-1}-S)^{2}\left(\sum_{k=2}^{n}\rho_{n}^{k-1}\widetilde{U}_{k-1}\right)^{2}.
\end{align}
Thus, we have 
\begin{align}
\mathcal{U}_{3,n}=\  & \left(n^{-3/2}\sum_{k=2}^{n}\widetilde{U}_{k-1}\right)^{2}\left(n^{-1}\sum_{i=1}^{n}\left(\frac{n(Q-S\rho^{i-1})}{nQ-S^{2}}\right)^{2}\widehat{U}_{i}^{2}\right)\nonumber \\
 & +\left(n^{-3/2}\sum_{k=2}^{n}\widetilde{U}_{k-1}\right)\left(n^{-3/2}\sum_{k=2}^{n}\rho_{n}^{k-1}\widetilde{U}_{k-1}\right)\left(2n^{-1}\sum_{i=1}^{n}\frac{n^{2}(Q-S\rho^{i-1})(n\rho^{i-1}-S)}{(nQ-S^{2})^{2}}\widehat{U}_{i}^{2}\right)\nonumber \\
 & +\left(n^{-3/2}\sum_{k=2}^{n}\rho_{n}^{k-1}\widetilde{U}_{k-1}\right)^{2}\left(n^{-1}\sum_{i=1}^{n}\left(\frac{n(n\rho^{i-1}-S)}{nQ-S^{2}}\right)^{2}\widehat{U}_{i}^{2}\right)\nonumber \\
=\  & \left(n^{-3/2}\sum_{k=2}^{n}\widetilde{U}_{k-1}\right)^{2}\left(n^{-1}\sum_{i=1}^{n}\left(\frac{n(Q-S\rho^{i-1})}{nQ-S^{2}}\right)^{2}U_{i}^{2}\right)\nonumber \\
 & +\left(n^{-3/2}\sum_{k=2}^{n}\widetilde{U}_{k-1}\right)\left(n^{-3/2}\sum_{k=2}^{n}\rho_{n}^{k-1}\widetilde{U}_{k-1}\right)\left(2n^{-1}\sum_{i=1}^{n}\frac{n^{2}(Q-S\rho^{i-1})(n\rho^{i-1}-S)}{(nQ-S^{2})^{2}}U_{i}^{2}\right)\nonumber \\
 & +\left(n^{-3/2}\sum_{k=2}^{n}\rho_{n}^{k-1}\widetilde{U}_{k-1}\right)^{2}\left(n^{-1}\sum_{i=1}^{n}\left(\frac{n(n\rho^{i-1}-S)}{nQ-S^{2}}\right)^{2}U_{i}^{2}\right)+o_{p}(1)\nonumber \\
\dto\  & \frac{h(1-e^{-2h})}{h(1-e^{-2h})-2(1-e^{-h})^{2}}\left(\int_{0}^{1}I_{h}(r)dr\right)^{2}\nonumber \\
 & -\frac{4h(1-e^{-h})}{h(1-e^{-2h})-2(1-e^{-h})^{2}}\left(\int_{0}^{1}I_{h}(r)dr\right)\left(\int_{0}^{1}e^{-hr}I_{h}(r)dr\right)\nonumber \\
 & +\frac{2h^{2}}{h(1-e^{-2h})-2(1-e^{-h})^{2}}\left(\int_{0}^{1}e^{-hr}I_{h}(r)dr\right)^{2},\label{eq:U3n}
\end{align}
where the second equality holds by (\ref{eq:Ui-relationships}) and
the convergence holds by Lemma \ref{lem:algebraic-relationship1}\ref{enu:nQS4}--\ref{enu:nQS6}
and Lemma \ref{lem:FCLT}\ref{enu:lem-c}, \ref{enu:lem-f}, and \ref{enu:lem-g}.

Combining (\ref{eq:U1n}), (\ref{eq:U2n}), and (\ref{eq:U3n}), we
have 
\begin{align}
 & \mathcal{U}_{1,n}-2\mathcal{U}_{2,n}+\mathcal{U}_{3n}\nonumber \\
\dto\  & \int_{0}^{1}I_{h}(r)^{2}dr-\frac{h(1-e^{-2h})}{h(1-e^{-2h})-2(1-e^{-h})^{2}}\left(\int_{0}^{1}I_{h}(r)dr\right)^{2}\nonumber \\
 & +\frac{4h(1-e^{-h})}{h(1-e^{-2h})-2(1-e^{-h})^{2}}\left(\int_{0}^{1}I_{h}(r)dr\right)\left(\int_{0}^{1}e^{-hr}I_{h}(r)dr\right)\nonumber \\
 & -\frac{2h^{2}}{h(1-e^{-2h})-2(1-e^{-h})^{2}}\left(\int_{0}^{1}e^{-hr}I_{h}(r)dr\right)^{2}\nonumber \\
=\  & \int_{0}^{1}I_{f,h}^{2}(r)dr.
\end{align}
Thus, we have
\begin{equation}
n^{-2}X_{1}'M_{X_{2}(\rho_{n})}\Delta^{2}M_{X_{2}(\rho_{n})}X_{1}\dto\int_{0}^{1}I_{f,h}^{2}(r)dr.\label{eq:meat-asy}
\end{equation}
Putting (\ref{eq:bread-asy}), (\ref{eq:numerator}), and (\ref{eq:meat-asy})
together gives 
\begin{equation}
T_{n}(\rho_{n})\dto\frac{\int_{0}^{1}I_{f,h}(r)dW(r)}{\left(\int_{0}^{1}I_{f,h}(r)^{2}dr\right)^{1/2}}.
\end{equation}

\medskip{}

\subsubsection{Close-to-Unit-Root/Unit-Root Case}\label{subsec:Close-to-Unit-Root/Unit-Root-Cas}

Next, we consider the case $h=0$. We rewrite (\ref{eq:bread-asy})
as
\begin{align}
 & n^{-2}X_{1}'M_{X_{2}(\rho_{n})}X_{1}\nonumber \\
=\  & n^{-2}(M_{X_{2}(\rho_{n})}\widetilde{U}_{-1})'\widetilde{U}_{-1}\nonumber \\
=\  & n^{-2}\sum_{i=2}^{n}\widetilde{U}_{i-1}\Bigg(\widetilde{U}_{i-1}-\frac{1}{nQ-S^{2}}\nonumber \\
 & \times\sum_{k=2}^{n}\Bigg(Q-2S+n-n(1-\rho_{n})(n-S)\left(\frac{1-\rho_{n}^{i-1}}{n(1-\rho_{n})}+\frac{1-\rho_{n}^{k-1}}{n(1-\rho_{n})}\right)\nonumber \\
 & +n^{3}(1-\rho_{n})^{2}\frac{1-\rho_{n}^{i-1}}{n(1-\rho_{n})}\frac{1-\rho_{n}^{k-1}}{n(1-\rho_{n})}\Bigg)\widetilde{U}_{k-1}\Bigg)\nonumber \\
=\  & n^{-2}\sum_{i=2}^{n}\widetilde{U}_{i-1}^{2}-n^{-2}\frac{1}{nQ-S^{2}}\sum_{i=2}^{n}\widetilde{U}_{i-1}\nonumber \\
 & \times\sum_{k=2}^{n}\Bigg(Q-2S+n-n(1-\rho_{n})(n-S)\left(\frac{1-\rho_{n}^{i-1}}{n(1-\rho_{n})}+\frac{1-\rho_{n}^{k-1}}{n(1-\rho_{n})}\right)\nonumber \\
 & +n^{3}(1-\rho_{n})^{2}\frac{1-\rho_{n}^{i-1}}{n(1-\rho_{n})}\frac{1-\rho_{n}^{k-1}}{n(1-\rho_{n})}\Bigg)\widetilde{U}_{k-1}\nonumber \\
=\  & n^{-2}\sum_{i=2}^{n}\widetilde{U}_{i-1}^{2}-\frac{n(Q-2S+n)}{nQ-S^{2}}\left(n^{-3/2}\sum_{i=2}^{n}\widetilde{U}_{i-1}\right)^{2}\nonumber \\
 & +\frac{2n^{2}(1-\rho_{n})(n-S)}{nQ-S^{2}}\left(n^{-3/2}\sum_{i=2}^{n}\widetilde{U}_{i-1}\right)\left(n^{-3/2}\sum_{i=2}^{n}\frac{1-\rho_{n}^{i-1}}{n(1-\rho_{n})}\widetilde{U}_{i-1}\right)\nonumber \\
 & -\frac{n^{4}(1-\rho_{n})^{2}}{nQ-S^{2}}\left(n^{-3/2}\sum_{i=2}^{n}\frac{1-\rho_{n}^{i-1}}{n(1-\rho_{n})}\widetilde{U}_{i-1}\right)^{2}\nonumber \\
\dto\  & \int_{0}^{1}W^{2}(r)dr-4\left(\int_{0}^{1}W(r)dr\right)^{2}+12\left(\int_{0}^{1}W(r)dr\right)\left(\int_{0}^{1}rW(r)dr\right)\nonumber \\
 & -12\left(\int_{0}^{1}rW(r)dr\right)^{2}\nonumber \\
=\  & \int_{0}^{1}W_{f}^{2}(r)dr,\label{eq:bread-asyh0}
\end{align}
where the convergence holds by Lemma \ref{lem:algebraic-relationsip2}\ref{enu:al-relation1}--\ref{enu:al-relation4}.
Observe that 
\begin{equation}
n^{-1}X_{1}'M_{X_{2}(\rho_{n})}U\ \text{and}\ n^{-2}X_{1}'M_{X_{2}(\rho_{n})}\Delta^{2}M_{X_{2}(\rho_{n})}X_{1}
\end{equation}
can be treated analogously by rewriting (\ref{eq:numerator}) and
(\ref{eq:meat}) in the form of (\ref{eq:bread-asyh0}). Therefore,
we have 
\begin{align}
n^{-1}X_{1}'M_{X_{2}(\rho_{n})}U & \stackrel{d}{\to}\int_{0}^{1}W_{f}(r)dW(r)\text{ and }\nonumber \\
n^{-2}X_{1}'M_{X_{2}(\rho_{n})}\Delta^{2}M_{X_{2}(\rho_{n})}X_{1} & \stackrel{d}{\to}\int_{0}^{1}W_{f}^{2}(r)dr.\label{eq:meat-asyh0}
\end{align}
Combining (\ref{eq:bread-asyh0})--(\ref{eq:meat-asyh0}), we have
\begin{equation}
T_{n}(\rho_{n})\dto\frac{\int_{0}^{1}W_{f}(r)dW(r)}{\left(\int_{0}^{1}W_{f}^{2}(r)dr\right)^{1/2}}
\end{equation}
when $h=0$.
\begin{rem}
Our proof for $h=0$ differs from that of Theorem 1 in AG12, where
it is necessary to show that the influence of $Y_{0}^{*}$ is negligible.
In contrast, we eliminate the influence of $Y_{0}^{*}$ via the projection
approach and thus avoid imposing any conditions on $Y_{0}^{*}$.
\end{rem}

\subsubsection{Stationary Case}

It remains to analyze the case $h=\infty$. It suffices to consider
the two sub-cases $\rho_{n}\to\rho^{*}<1$ and $\rho_{n}\to1$. First,
assume that $\rho_{n}\to1$ and $n(1-\rho_{n})\to\infty$. We define
\begin{equation}
a_{n}=n^{1/2}\frac{\E\widetilde{U}_{0}^{2}}{\left(\E(\widetilde{U}_{0}^{2}U_{1}^{2})\right)^{1/2}}\quad\text{and}\quad d_{n}=\frac{\E\widetilde{U}_{0}^{2}}{\left(\E(\widetilde{U}_{0}^{2}U_{1}^{2})\right)^{1/2}}.
\end{equation}
We redefine $X_{2}(\rho_{n})$ as an $n\times2$ matrix whose $i$-th
row is $(1,h_{n}^{1/2}\rho_{n}^{i-1})$. This redefinition does not
affect $P_{X_{2}(\rho_{n})}$ or $M_{X_{2}(\rho_{n})}$, but simplifies
the analysis. By algebra, we have
\begin{align}
a_{n}(\widehat{\rho}_{n}(\rho_{n})-\rho_{n}) & =\left(n^{-1}\frac{X_{1}'M_{X_{2}(\rho_{n})}X_{1}}{\E\widetilde{U}_{0}^{2}}\right)^{-1}\frac{n^{-1/2}X_{1}'M_{X_{2}(\rho_{n})}U}{\left(\E(\widetilde{U}_{0}^{2}U_{1}^{2})\right)^{1/2}}\nonumber \\
 & =\left(n^{-1}\frac{\widetilde{U}_{-1}'M_{X_{2}(\rho_{n})}\widetilde{U}_{-1}}{\E\widetilde{U}_{0}^{2}}\right)^{-1}\frac{n^{-1/2}\widetilde{U}_{-1}'M_{X_{2}(\rho_{n})}U}{\left(\E(\widetilde{U}_{0}^{2}U_{1}^{2})\right)^{1/2}}=:\nu_{n}\xi_{n},\label{eq:rhon-rho-decomposition}
\end{align}
where $\nu_{n}$ and $\xi_{n}$ are defined implicitly by the last
equality. Furthermore, we define the martingale difference sequence
(MDS)
\begin{equation}
\zeta_{i}:=n^{-1/2}\frac{\widetilde{U}_{i-1}U_{i}}{\left(\E(\widetilde{U}_{0}^{2}U_{1}^{2})\right)^{1/2}}.
\end{equation}
We now show that 
\begin{equation}
n^{-1/2}\frac{\widetilde{U}_{-1}'P_{X_{2}(\rho_{n})}U}{\left(\E(\widetilde{U}_{0}^{2}U_{1}^{2})\right)^{1/2}}\pto0\quad\text{and}\quad\sum_{i=2}^{n}\zeta_{i}\dto N(0,1),\label{eq:clt-h-0}
\end{equation}
which establish that
\begin{equation}
\xi_{n}=\sum_{i=2}^{n}\zeta_{i}-n^{-1/2}\frac{\widetilde{U}_{-1}'P_{X_{2}(\rho_{n})}U}{\left(\E(\widetilde{U}_{0}^{2}U_{1}^{2})\right)^{1/2}}\dto N(0,1).
\end{equation}

To show the first result in (\ref{eq:clt-h-0}), note that
\begin{equation}
n^{-1/2}X_{2}(\rho_{n})'U=\left(n^{-1/2}\sum_{i=1}^{n}U_{i},n^{-1/2}h_{n}^{1/2}\sum_{i=1}^{n}\rho_{n}^{i-1}U_{i}\right)'=(O_{p}(1),O_{p}(1))'\label{eq:X2U}
\end{equation}
by the central limit theorem (CLT) for a triangular array of MDS $U_{i}$
and $h_{n}^{1/2}\rho_{n}^{i-1}U_{i}$. Note that
\begin{align}
\frac{h_{n}^{1/2}(1-\rho_{n}^{n})}{n(1-\rho_{n})} & =\frac{1}{h_{n}^{1/2}}+o(h_{n}^{-1/2})=o(1)\text{ and}\nonumber \\
\frac{h_{n}(1-\rho_{n}^{2n})}{n(1-\rho_{n}^{2})} & =\frac{1}{2}+o(1).\label{eq:X2X2}
\end{align}
Hence, we have
\begin{equation}
(n^{-1}X_{2}(\rho_{n})'X_{2}(\rho_{n}))^{-1}=\diag(1,2)+o(1)=O(1).\label{eq:X2X2_inv}
\end{equation}
Finally, we have
\begin{equation}
n^{-1}(1-\rho_{n})^{1/2}\widetilde{U}_{-1}'X_{2}(\rho_{n})=(o_{p}(1),o_{p}(1)),\label{eq:U-1X2}
\end{equation}
by Lemma \ref{lem:case2lem}\ref{enu:case2lem-a}. Then, it follows
that 
\begin{align}
 & n^{-1/2}\frac{\widetilde{U}_{-1}'P_{X_{2}(\rho_{n})}U}{\left(\E(\widetilde{U}_{0}^{2}U_{1}^{2})\right)^{1/2}}\nonumber \\
= & \left(n^{-1/2}X_{2}(\rho_{n})'U\right)\left(n^{-1}X_{2}(\rho_{n})'X_{2}(\rho_{n})\right)^{-1}\left(n^{-1}\dfrac{\widetilde{U}_{-1}'X_{2}(\rho_{n})}{\left(\E(\widetilde{U}_{0}^{2}U_{1}^{2})\right)^{1/2}}\right)\pto0,
\end{align}
where the convergence holds by (\ref{eq:X2U})--(\ref{eq:U-1X2})
and the fact that $\E(\widetilde{U}_{0}^{2}U_{1}^{2})=O((1-\rho_{n})^{-1})$
by Lemma \ref{lem:Order-expect}.

To establish the second result in (\ref{eq:clt-h-0}), we adapt the
proof of Lemma 1 in \textcite{giraitis2006uniform}. It suffices to
verify the analogues of (11) and (12) in \textcite{giraitis2006uniform};
that is, to show the Lindeberg condition 
\begin{equation}
\sum_{i=2}^{n}\E(\zeta_{i}^{2}1(|\zeta_{i}|>\delta)|\mathcal{G}_{i-1})\pto0\ \ \text{for any }\delta>0,
\end{equation}
and
\begin{equation}
\sum_{i=2}^{n}\E(\zeta_{i}^{2}|\mathcal{G}_{i-1})\pto1.
\end{equation}
The former follows from Lemma \ref{lem:integrable-condi}, while the
latter is shown in Lemma \ref{lem:moment-condition}, using the fact
that $\sum_{i=2}^{n}\E\zeta_{i}^{2}\to1.$ Thus, we have proved 
\begin{equation}
\xi_{n}\dto N\left(0,1\right).\label{eq:numerator_xi}
\end{equation}

Next, we show $\nu_{n}\pto1$. By Lemma \ref{lem:case2lem}\ref{enu:case2lem-b}
and Lemma \ref{lem:Order-expect}, we have
\begin{equation}
\frac{n^{-1}\widetilde{U}_{-1}'\widetilde{U}_{-1}}{\E\widetilde{U}_{0}^{2}}\pto1\quad\text{and}\quad\frac{n^{-1}\widetilde{U}_{-1}'P_{X_{2}(\rho_{n})}\widetilde{U}_{-1}}{\E\widetilde{U}_{0}^{2}}\pto0,\label{eq:convergence-Utilde-square}
\end{equation}
which imply $\nu_{n}\pto1$. Together with (\ref{eq:numerator_xi}),
we have 
\begin{equation}
a_{n}(\widehat{\rho}_{n}(\rho_{n})-\rho_{n})\dto N\left(0,1\right).
\end{equation}

Finally, we show that $d_{n}\widehat{\sigma}_{n}\pto1$, which implies
that the $t$-statistic $T_{n}\left(\rho_{n}\right)\dto N\left(0,1\right)$
for the sub-case $\rho_{n}\to1$ and $n(1-\rho_{n})\to\infty$. By
(\ref{eq:convergence-Utilde-square}), it is sufficient to show that
\begin{equation}
\frac{n^{-1}X_{1}'M_{X_{2}(\rho_{n})}\Delta^{2}M_{X_{2}(\rho_{n})}X_{1}}{\E(\widetilde{U}_{0}^{2}U_{1}^{2})}\pto1.\label{eq:meat-convergence}
\end{equation}
By Lemma \ref{lem:case2lem}\ref{enu:case2lem-e}--\ref{enu:case2lem-g},
we have 
\begin{align}
(\E(\widetilde{U}_{0}^{2}U_{1}^{2}))^{-1}n^{-1}\widetilde{U}_{-1}'\Delta^{2}\widetilde{U}_{-1} & \pto1,\nonumber \\
(1-\rho_{n})^{1/2}n^{-1}(X_{2}(\rho_{n})'\Delta^{2}\widetilde{U}_{-1}) & =O_{p}(1),\text{ and}\nonumber \\
n^{-1}(X_{2}(\rho_{n})'\Delta^{2}X_{2}(\rho_{n})) & =O_{p}(1).
\end{align}
These results combined with Lemma \ref{lem:Order-expect}, (\ref{eq:X2X2_inv}),
and (\ref{eq:U-1X2}) imply (\ref{eq:meat-convergence}).

\medskip{}

Next, we consider the case where $\rho_{n}\to\rho^{*}<1$. We replicate
the steps used in the proof for the case when $n(1-\rho_{n})\to\infty$
and $\rho_{n}\to1$. Equation (\ref{eq:X2U}) holds by the CLT for
MDS applied to $U_{i}$ and $h_{n}^{1/2}\rho_{n}^{i-1}U_{i}$, and
(\ref{eq:X2X2}) holds because $h_{n}=n(1-\rho_{n})\to\infty$. Following
the same proof strategy as in Lemma 8 in AG12, Lemma \ref{lem:case2lem}
also holds when $\rho_{n}\to\rho^{*}<1$. In this case, the bounds
involving $1-\rho_{n}$ can be replaced by $1-\rho^{*}$, which is
a bounded constant. Therefore, (\ref{eq:U-1X2}) still holds. Note
that $\E(\widetilde{U}_{0}^{2}U_{1}^{2})=O(1)$ in this case. Hence,
we have 
\begin{equation}
n^{-1/2}\frac{\widetilde{U}_{-1}'P_{X_{2}(\rho_{n})}U}{\left(\E(\widetilde{U}_{0}^{2}U_{1}^{2})\right)^{1/2}}\pto0.
\end{equation}
Additionally, Lemmas \ref{lem:moment-condition} and \ref{lem:integrable-condi}
also hold when $\rho_{n}\to\rho^{*}<1$, hence we have $\sum_{i=1}^{n}\zeta_{i}\dto N(0,1)$.

To establish (\ref{eq:convergence-Utilde-square}) and (\ref{eq:meat-convergence}),
it suffices to apply the weak law of large numbers (WLLN) for triangular
arrays of mean zero, $L^{1+\delta}$ bounded, and $L^{1}$-near-epoch
dependent ($L^{1}$-NED) random variables. \textcite[p.464]{andrews1988laws}
shows that the latter conditions imply that the array is a uniformly
integrable $L^{1}$ mixingale for which the WLLN holds; see \textcite[Theorem 2]{andrews1988laws}.
For example, we verify that $\widetilde{U}_{i}^{2}-\E\widetilde{U}_{0}^{2}$
is $L^{1}$-NED with respect to the $\sigma$-field $\mathcal{G}_{i}$,
which implies $n^{-1}\widetilde{U}_{-1}'\widetilde{U}_{-1}-\E\widetilde{U}_{0}^{2}\pto0$
by the WLLN.

By definition of $L^{1}$-NED in \textcite[p.464]{andrews1988laws},
we need to show that there exist constants $\{d_{i}:1\leq i\leq k_{n},n\geq1\}$
and $\{v_{m}:m\geq0\}$ such that $v_{m}\down0$ as $m\to\infty$
and that
\begin{equation}
\E|\widetilde{U}_{i}^{2}-\E(\widetilde{U}_{i}^{2}|\mathcal{G}_{i-m,i+m})|\leq d_{i}v_{m},\label{eq:L1-NED}
\end{equation}
where $\mathcal{G}_{i-m,i+m}$ is a $\sigma$-filed such that $U_{j}\in\mathcal{G}_{i-m,i+m}$
for $i-m\leq j\le i+m$. Note that for $m\geq i-1$, we have $\E(\widetilde{U}_{i}^{2}|\mathcal{G}_{i-m,i+m})=\widetilde{U}_{i}^{2}$,
so that $d_{i}v_{m}=0$. Therefore, we only need to consider the case
$m<i-1$: 
\begin{align}
\E(\widetilde{U}_{i}^{2}|\mathcal{G}_{i-m,i+m}) & =\E\left(\left(\sum_{j=i-m}^{i}\rho_{n}^{i-j}U_{j}+\rho_{n}^{m+1}\widetilde{U}_{i-m-1}\right)^{2}|\mathcal{G}_{i-m,i+m}\right)\nonumber \\
 & =\left(\sum_{j=i-m}^{i}\rho_{n}^{i-j}U_{j}\right)^{2}+\rho_{n}^{2m+2}\E\widetilde{U}_{i-m-1}^{2}.\label{eq:condl_mean_U_itilde_sq}
\end{align}
Using (\ref{eq:condl_mean_U_itilde_sq}), we have
\begin{align}
 & \E|\widetilde{U}_{i}^{2}-\E(\widetilde{U}_{i}^{2}|\mathcal{G}_{i-m,i+m})|\nonumber \\
=\  & \E\left[\left|\rho_{n}^{2m+2}(\widetilde{U}_{i-m-1}^{2}-\E\widetilde{U}_{i-m-1}^{2})+2\rho_{n}^{m+1}\widetilde{U}_{i-m-1}\sum_{j=i-m}^{i}\rho_{n}^{i-j}U_{j}\right|\right]\nonumber \\
\leq\  & 2\rho_{n}^{2m+2}\sum_{j=1}^{i-m-1}\sum_{k=1}^{i-m-1}\rho_{n}^{2(i-m-1)-j-k}\E\left|U_{j}U_{k}\right|+2\rho_{n}^{m+1}\sum_{k=1}^{i-m-1}\sum_{j=i-m}^{i}\rho_{n}^{2i-m-1-j-k}\E|U_{j}U_{k}|\nonumber \\
\leq\  & 2M(1-\rho_{n})^{-2}(\rho_{n}^{2m+2}+\rho_{n}^{m+1})\to0,
\end{align}
where the first inequality holds by the triangle inequality and the
second inequality holds by $\E|U_{j}U_{k}|\leq M$ using the moment
condition in part \ref{enu:para-space4} of the parameter space $\Lambda_{n}$.
Therefore, (\ref{eq:L1-NED}) holds by taking $d_{i}=2M(1-\rho_{n})^{-2}$
and $v_{m}=2((\rho^{*})^{2m+2}+(\rho^{*})^{m+1})\down0$ as $m\to\infty$
because $\rho_{n}\to\rho^{*}<1$.

\subsection{Proofs of Lemmas}
\begin{proof}[Proof of Lemma \ref{lem:continuous-Ifh}]
 By a Taylor expansion of $e^{-x}$ around $x=0$, we have 
\begin{equation}
e^{-x}=1-x+\frac{1}{2}x^{2}-\frac{1}{6}x^{3}+O(x^{4}).
\end{equation}
Hence, 
\begin{align}
e^{-2h} & =1-2h+2h^{2}-\frac{4}{3}h^{3}+O(h^{4})\quad\text{and}\nonumber \\
e^{-hr} & =1-hr+\frac{1}{2}h^{2}r^{2}-\frac{1}{6}h^{3}r^{3}+O(h^{4})\text{ as \ensuremath{h\to0}},
\end{align}
where the second equality holds uniformly for $r\in[0,1]$. Substituting
these equations into the denominator $h(1-e^{-2h})-2(1-e^{-h})^{2}$
of $\alpha_{h}(r)$ and $\beta_{h}(r)$ gives
\begin{align}
h(1-e^{-2h})-2(1-e^{-h})^{2}=\  & h(2h-2h^{2}+\frac{4}{3}h^{3}+O(h^{4}))-2\left(h-\frac{1}{2}h^{2}+\frac{1}{6}h^{3}+O(h^{4})\right)^{2}\nonumber \\
=\  & 2h^{2}-2h^{3}+\frac{4}{3}h^{4}-2h^{2}\left(1-h+\frac{7}{12}h^{2}\right)+O(h^{5})\nonumber \\
=\  & \frac{1}{6}h^{4}+O(h^{5})\text{ as \ensuremath{h\to0}.}
\end{align}
Similarly, for the numerators of $\alpha_{h}(r)$ and $\beta_{h}(r)$,
we have 
\begin{align}
h[1-e^{-2h}-2(1-e^{-h})e^{-hr}+2he^{-hr}-2(1-e^{-h})] & =\frac{2-3r}{3}h^{4}+O(h^{5})\quad\text{and}\nonumber \\
2h[he^{-hr}-(1-e^{-h})] & =(2r-1)h^{4}+O(h^{5})\text{ as \ensuremath{h\to0}}.
\end{align}
Therefore, we obtain
\begin{align}
\alpha_{h}(r) & =\frac{\frac{2-3r}{3}h^{4}+O(h^{5})}{\frac{1}{6}h^{4}+O(h^{5})}=4-6r+O(h)\quad\text{and}\nonumber \\
\beta_{h}(r) & =\frac{(2r-1)h^{4}+O(h^{5})}{\frac{1}{6}h^{4}+O(h^{5})}=12r-6+O(h)\text{ as \ensuremath{h\to0}},
\end{align}
uniformly over $r\in[0,1]$. Letting $h\down0$ completes the proof.
\end{proof}
\begin{proof}[Proof of Lemma \ref{lem:algebraic-relationship1}]
As $\rho_{n}=1-h_{n}/n$ and $h_{n}\to h\in(0,\infty)$, (\ref{eq:QS-expansion})
gives the proofs of parts \ref{enu:newnQ} and \ref{enu:newnS} of
Lemma \ref{lem:algebraic-relationship1}:
\begin{align}
Q & =\frac{1-(1-n^{-1}h_{n})^{2n}}{2h_{n}/n-h_{n}^{2}/n^{2}}=\frac{1-e^{-2h_{n}}+o(1)}{2h_{n}/n+o(n^{-1})}=\frac{1-e^{-2h_{n}}}{2h_{n}}n+o(n)\text{ and}\nonumber \\
S & =\frac{1-(1-n^{-1}h_{n})^{n}}{h_{n}/n}=\frac{1-e^{-h_{n}}+o(1)}{h_{n}/n}=\frac{1-e^{-h_{n}}}{h_{n}}n+o(n).\label{eq:QS-expansion}
\end{align}
Then, we have 
\begin{align}
\frac{n^{2}}{nQ-S^{2}} & =\frac{n^{2}}{n^{2}(1-e^{-2h_{n}})/2h_{n}-n^{2}(1-e^{-h_{n}})^{2}/h_{n}^{2}+o(n^{2})}\nonumber \\
 & =\frac{1}{(1-e^{-2h_{n}})/2h_{n}-(1-e^{-h_{n}})^{2}/h_{n}^{2}}+o(1)\nonumber \\
 & \to\frac{2h^{2}}{h(1-e^{-2h})-2(1-e^{-h})^{2}}.
\end{align}
 Similarly, we compute 
\begin{align}
\frac{nS}{nQ-S^{2}} & =\frac{n^{2}(1-e^{-h_{n}})/h_{n}+o(n^{2})}{n^{2}(1-e^{-2h_{n}})/2h_{n}-n^{2}(1-e^{-h_{n}})^{2}/h_{n}^{2}+o(n^{2})}\nonumber \\
 & =\frac{(1-e^{-h_{n}})/h_{n}}{(1-e^{-2h_{n}})/2h_{n}-(1-e^{-h_{n}})^{2}/h_{n}^{2}}+o(1)\nonumber \\
 & \to\frac{2h(1-e^{-h})}{h(1-e^{-2h})-2(1-e^{-h})^{2}}
\end{align}
and 
\begin{align}
\frac{nQ}{nQ-S^{2}} & =\frac{n^{2}(1-e^{-2h_{n}})/2h_{n}+o(n^{2})}{n^{2}(1-e^{-2h_{n}})/2h_{n}-n^{2}(1-e^{-h_{n}})^{2}/h_{n}^{2}+o(n^{2})}\nonumber \\
 & =\frac{n^{2}(1-e^{-2h_{n}})/2h_{n}}{n^{2}(1-e^{-2h_{n}})/2h_{n}-n^{2}(1-e^{-h_{n}})^{2}/h_{n}^{2}}+o(1)\nonumber \\
 & \to\frac{h(1-e^{-2h})}{h(1-e^{-2h})-2(1-e^{-h})^{2}}.
\end{align}

For part \ref{enu:nQS4}, note that 
\begin{align}
 & n^{-1}\sum_{i=1}^{n}\left(\frac{n(Q-S\rho^{i-1})}{nQ-S^{2}}\right)^{2}\nonumber \\
=\  & n^{-1}\sum_{i=1}^{n}\left(\frac{nQ}{nQ-S^{2}}\right)^{2}+n^{-1}\sum_{i=1}^{n}\left(\frac{nS}{nQ-S^{2}}\right)^{2}\rho_{n}^{2i-2}-2n^{-1}\sum_{i=1}^{n}\frac{n^{2}SQ}{(nQ-S^{2})^{2}}\rho_{n}^{i-1}\nonumber \\
=\  & \left(\frac{nQ}{nQ-S^{2}}\right)^{2}+\left(\frac{nS}{nQ-S^{2}}\right)^{2}n^{-1}Q-2\frac{n^{2}SQ}{(nQ-S^{2})^{2}}n^{-1}S\nonumber \\
=\  & \frac{nQ}{nQ-S^{2}}.
\end{align}
Furthermore, we have 
\begin{align}
 & n^{-1}\sum_{i=1}^{n}\frac{n^{2}(Q-S\rho^{i-1})(n\rho^{i-1}-S)}{(nQ-S^{2})^{2}}\nonumber \\
=\  & n^{-1}\sum_{i=1}^{n}\frac{nQ}{nQ-S^{2}}\frac{n^{2}}{nQ-S^{2}}\rho_{n}^{i-1}-n^{-1}\sum_{i=1}^{n}\frac{nQ}{nQ-S^{2}}\frac{nS}{nQ-S^{2}}\nonumber \\
 & -n^{-1}\sum_{i=1}^{n}\frac{nS}{nQ-S^{2}}\frac{n^{2}}{nQ-S^{2}}\rho_{n}^{2i-2}+n^{-1}\sum_{i=1}^{n}\left(\frac{nS}{nQ-S^{2}}\right)^{2}\rho_{n}^{i-1}\nonumber \\
=\  & \frac{nQ}{nQ-S^{2}}\frac{n^{2}}{nQ-S^{2}}n^{-1}S-\frac{nQ}{nQ-S^{2}}\frac{nS}{nQ-S^{2}}\nonumber \\
 & -\frac{nS}{nQ-S^{2}}\frac{n^{2}}{nQ-S^{2}}n^{-1}Q+\left(\frac{nS}{nQ-S^{2}}\right)^{2}n^{-1}S\nonumber \\
=\  & \frac{-nS}{(nQ-S^{2})}
\end{align}
and 
\begin{align}
 & n^{-1}\sum_{i=1}^{n}\left(\frac{n(n\rho^{i-1}-S)}{nQ-S^{2}}\right)^{2}\nonumber \\
=\  & n^{-1}\sum_{i=1}^{n}\frac{n^{4}}{(nQ-S^{2})^{2}}\rho_{n}^{2i-2}-2n^{-1}\sum_{i=1}^{n}\frac{n^{3}S}{(nQ-S^{2})^{2}}\rho_{n}^{i-1}+n^{-1}\sum_{i=1}^{n}\frac{n^{2}S^{2}}{(nQ-S^{2})^{2}}\nonumber \\
=\  & \frac{n^{4}}{(nQ-S^{2})^{2}}n^{-1}Q-2\frac{n^{3}S}{(nQ-S^{2})^{2}}n^{-1}S+\frac{n^{2}S^{2}}{(nQ-S^{2})^{2}}\nonumber \\
=\  & \frac{n^{2}}{nQ-S^{2}}.
\end{align}
This completes the proof.
\end{proof}
\begin{proof}[Proof of Lemma \ref{lem:algebraic-relationsip2}]
For part \ref{enu:al-relation1}, the result holds trivially for
$i=1$. When $i=[nr]>1$, letting $h_{n}=n(1-\rho_{n})\to0$, we have
\begin{equation}
\frac{1-\rho_{n}^{i-1}}{n(1-\rho_{n})}=\frac{1-(1-\frac{h_{n}}{n})^{i-1}}{h_{n}}=\frac{1-(1-(i-1)\frac{h_{n}}{n}+{i-1 \choose 2}\frac{h_{n}^{2}}{n^{2}})}{h_{n}}+o(1)\to r.
\end{equation}

For part \ref{enu:al-relation2}, by Lemma \ref{lem:algebraic-relationship1}\ref{enu:nQS1}--\ref{enu:nQS3},
we have
\begin{align}
\frac{n(Q-2S+n)}{nQ-S^{2}} & =\frac{h_{n}(1-e^{-2h_{n}})-4h_{n}(1-e^{-h_{n}})+2h_{n}^{2}}{h_{n}(1-e^{-2h_{n}})-2(1-e^{-h_{n}})^{2}}+o(1).
\end{align}
Using the same Taylor expansion argument as in the proof of Lemma
\ref{lem:continuous-Ifh}, we have 
\begin{align}
 & h_{n}(1-e^{-2h_{n}})-4h_{n}(1-e^{-h_{n}})+2h_{n}^{2}\nonumber \\
=\  & h_{n}(2h_{n}-2h_{n}^{2}+\frac{4}{3}h_{n}^{3}+O(h^{4}))-4h_{n}(h_{n}-\frac{1}{2}h_{n}^{2}+\frac{1}{6}h_{n}^{3}+O(h_{n}^{4}))+2h_{n}^{2}\nonumber \\
=\  & \frac{2}{3}h_{n}^{4}+O(h_{n}^{5}),
\end{align}
and 
\begin{equation}
h_{n}(1-e^{-2h_{n}})-2(1-e^{-h_{n}})^{2}=\frac{1}{6}h_{n}^{4}+O(h_{n}^{5}).
\end{equation}
Therefore, we derive
\begin{equation}
\frac{n(Q-2S+n)}{nQ-S^{2}}=\frac{\frac{2}{3}h_{n}^{4}+O(h_{n}^{5})}{\frac{1}{6}h_{n}^{4}+O(h_{n}^{5})}\to4.
\end{equation}

For part \ref{enu:al-relation3}, we have 
\begin{align}
\frac{n^{2}(1-\rho_{n})(n-S)}{nQ-S^{2}} & =\frac{h_{n}(2h_{n}^{2}-2h_{n}(1-e^{-h_{n}}))}{h_{n}(1-e^{-2h_{n}})-2(1-e^{-h_{n}})^{2}}+o(1)\nonumber \\
 & =\frac{2h_{n}^{3}-2h_{n}^{2}(h_{n}-\frac{1}{2}h_{n}^{2}+O(h_{n}^{3}))}{\frac{1}{6}h_{n}^{4}+O(h_{n}^{5})}+o(1)\to6.
\end{align}

For part \ref{enu:al-relation4}, we have 
\begin{align}
\frac{n^{4}(1-\rho_{n})^{2}}{nQ-S^{2}} & =\frac{2h_{n}^{4}}{h_{n}(1-e^{-2h_{n}})-2(1-e^{-h_{n}})^{2}}+o(1)\nonumber \\
 & =\frac{2h_{n}^{4}}{\frac{1}{6}h_{n}^{4}+O(h_{n}^{5})}+o(1)\to12.
\end{align}
Thus, the proof is completed..
\end{proof}
\begin{proof}[Proof of Lemma \ref{lem:FCLT}]
By setting $\phi_{n,i}=1$ in the proof of Lemma 5 in AG12, we have
already established this lemma except for parts \ref{enu:lem-e},
\ref{enu:lem-g}, \ref{enu:lem-i}, \ref{enu:lem-l}, and certain
combinations of part \ref{enu:lem-n}. We prove these remaining results
below.

For parts \ref{enu:lem-e}, \ref{enu:lem-g}, \ref{enu:lem-i}, and
\ref{enu:lem-l}, we follow the proof of Lemma 4 in AG12. We prove
part \ref{enu:lem-g} here, and the others follow an analogous argument.
When $h=0$, we can directly apply Theorem 4.4 of \textcite{hansen1992convergence}.
When $h>0$, we have 
\begin{equation}
n^{-1/2}\rho_{n}^{[nr]-1}\widetilde{U}_{[nr]}\Rightarrow e^{-hr}\int_{0}^{r}\exp(-(r-s)h)dW(s)=e^{-hr}I_{h}(r),
\end{equation}
by Theorem 4.4 of \textcite{hansen1992convergence}. Hence, 
\begin{equation}
n^{-3/2}\sum_{i=1}^{n}\rho_{n}^{i-1}\widetilde{U}_{n,i-1}\dto\int_{0}^{1}e^{-hr}I_{h}(r)dr
\end{equation}
by the continuous mapping theorem. Parts \ref{enu:lem-e}, \ref{enu:lem-i},
and \ref{enu:lem-l} can be proved in a similar way and we omit the
details.

For part \ref{enu:lem-n}, the case with $l_{0}=0$ is established
in Lemma 5 of AG12. The cases with $l_{0}=1$ and $2$ follow by the
same argument, as the rescaling of $(\rho_{n}^{i-1})^{l_{0}}$ leaves
the conclusion unchanged.
\end{proof}
\begin{proof}[Proof of Lemma \ref{lem:case2lem}]
First, note that
\begin{equation}
n^{-1}(1-\rho_{n})^{1/2}\widetilde{U}_{-1}'X_{2}(\rho_{n})=\left(n^{-3/2}h_{n}^{1/2}\sum_{i=2}^{n}\widetilde{U}_{i-1},n^{-3/2}h_{n}\sum_{i=2}^{n}\rho_{n}^{i-1}\widetilde{U}_{i-1}\right).\label{eq:U_-1X2}
\end{equation}
Also note that $n^{-3/2}h_{n}^{1/2}\sum_{i=1}^{n}\widetilde{U}_{i-1}=o_{p}(1)$
has been proved in Lemma 8 of AG12. For the second term in (\ref{eq:U_-1X2}),
by Markov's inequality, it suffices to show that 
\begin{equation}
n^{-3}h_{n}^{2}\E\left(\sum_{i=2}^{n}\rho_{n}^{i-1}\widetilde{U}_{i-1}\right)^{2}=o(1).
\end{equation}
Notice that 
\begin{align}
\E\left(\sum_{i=2}^{n}\rho_{n}^{i-1}\widetilde{U}_{i-1}\right)^{2} & =\E\left(\sum_{i=2}^{n}\sum_{k=1}^{i-1}\rho_{n}^{i+k-1}U_{i-1-k}\right)^{2}\nonumber \\
 & \leq\sum_{i,j=2}^{n}\sum_{k=1}^{i-1}\sum_{l=1}^{j-1}\rho_{n}^{i+j+k+l-2}\left|\E(U_{|i+l-(j+k)|+1}U_{1})\right|.\label{eq:rhontildeU-expansion}
\end{align}
By the stationarity of the $U_{i}$'s, the $|\E(U_{|i+l-(j+k)|+1}U_{1})|$
term in (\ref{eq:rhontildeU-expansion}) is uniformly bounded over
$i,j,k,l$, and hence, can be ignored.

The contribution of the summands with $i+l=j+k$ is 
\begin{align}
 & \sum_{i=2}^{n}\sum_{k=1}^{i-1}\rho_{n}^{2(i+k)-2}+2\sum_{i>j=2}^{n}\sum_{k=i-j+1}^{i-1}\rho_{n}^{2(j+k)-2}\nonumber \\
\leq\  & \sum_{i=2}^{\infty}\sum_{k=1}^{i-1}\rho_{n}^{2(i+k)-2}+2\sum_{i>j=2}^{n}\rho_{n}^{i+j}\frac{1-\rho_{n}^{2j}}{1-\rho_{n}^{2}}\nonumber \\
\leq\  & \left(\sum_{i=2}^{\infty}\rho_{n}^{2i-2}\right)^{2}+(1-\rho_{n})^{-1}\left(\sum_{i=2}^{\infty}\rho_{n}^{i}\right)^{2}=O((1-\rho_{n})^{-3}).
\end{align}
This term is negligible because 
\begin{equation}
(1-\rho_{n})^{-3}n^{-3}h_{n}^{2}=h_{n}^{-1}\to0.
\end{equation}

Next, we consider the case $i+l\neq j+k$. By symmetry, we only need
to prove it for $i+l>j+k$. We separate this case into two sub-cases
(i) $l>j+k-i$, $k\geq i-j+1$ and (ii) $l>j+k-i$, $k\leq i-j$.
By the properties of $\alpha$-mixing variables \parencite[Theorem 3, p. 9]{doukhan1995mixing}
and parts \ref{enu:para-space3} and \ref{enu:para-space4} of $\Lambda_{n}$,
we have 
\begin{equation}
\E(U_{l}U_{k})\leq8\|U_{l}\|_{\zeta}\|U_{k}\|_{\zeta}\alpha(|l-k|)^{1-2/\zeta}\leq8CM|l-k|^{-3(\zeta-2)/(\zeta-3)}\label{eq:alpha-mixing-inequalities}
\end{equation}
for $\zeta>3$. Write $\epsilon=-3/(\zeta-3)$. Hence, the sum for
the sub-case (i) can be bounded by 
\begin{align}
\  & 8CM\sum_{i>j=2}^{n}\sum_{k=1}^{i-1}\sum_{l=j+k-i+1}^{j-1}\rho_{n}^{i+j+k+l-2}(i+l-j-k)^{-3-\epsilon}\nonumber \\
=\  & 8CM\sum_{i>j=2}^{n}\sum_{k=1}^{i-1}\sum_{m=1}^{i-1-k}\rho_{n}^{2(j+k)+m-2}m^{-3-\epsilon}\nonumber \\
\leq\  & 8CM\sum_{i>j=2}^{n}\left(\sum_{k=1}^{i-1}\rho_{n}^{i+j+2k-2}\right)\left(\sum_{m=1}^{i-1-k}\rho_{n}^{m}m^{-3-\epsilon}\right)\nonumber \\
=\  & O((1-\rho_{n})^{-3}),
\end{align}
 where the last equality holds by 
\begin{equation}
\sum_{l=1}^{\infty}\rho_{n}^{2k}=O((1-\rho_{n}^{2})^{-1})=O((1-\rho_{n})^{-1})\text{ and }\sum_{m=1}^{\infty}\rho_{n}^{m}m^{-3-\epsilon}=O(1).
\end{equation}
Therefore, the summand with $j\neq k$ is also negligible because
\begin{equation}
(1-\rho_{n})^{-3}n^{-3}h_{n}^{2}=h_{n}^{-1}\to0.
\end{equation}
Similarly, we can prove the sub-case (ii), which is also $O((1-\rho_{n})^{-3})$.

Parts \ref{enu:case2lem-b} and \ref{enu:case2lem-c} have been proved
in Lemma 8 of AG12.

For part \ref{enu:case2lem-d}, note that 
\begin{equation}
(n^{-1}X(\rho_{n})'X(\rho_{n}))^{-1}n^{-1}X_{2}(\rho_{n})'U=(\det)^{-1}(T_{1},T_{2},T_{3})',
\end{equation}
where 
\begin{align}
\det:=\  & h_{n}\Bigg[n^{-2}(nQ-S^{2})\left(n^{-1}\sum_{i=2}^{n}\widetilde{U}_{i-1}^{2}\right)-n^{-1}\left(n^{-1}\sum_{i=2}^{n}\rho_{n}^{i-1}\widetilde{U}_{i-1}\right)^{2}\nonumber \\
 & +2n^{-1}S\left(n^{-1}\sum_{i=2}^{n}\widetilde{U}_{i-1}^{2}\right)\left(n^{-1}\sum_{i=1}^{n}\rho_{n}^{i-1}\widetilde{U}_{i-1}\right)-n^{-1}Q\left(n^{-1}\sum_{i=2}^{n}\widetilde{U}_{i-1}^{2}\right)^{2}\Bigg],\nonumber \\
T_{1}:=\  & h_{n}\Bigg[\left(n^{-1}Sn^{-1}\sum_{i=2}^{n}\rho_{n}^{i-1}\widetilde{U}_{i-1}-n^{-1}Qn^{-1}\sum_{i=2}^{n}\widetilde{U}_{i-1}\right)\left(n^{-1}\sum_{i=1}^{n}U_{i}\right)\nonumber \\
 & -\left(n^{-1}\sum_{i=2}^{n}\rho_{n}^{i-1}\widetilde{U}_{i-1}-n^{-1}Sn^{-1}\sum_{i=2}^{n}\widetilde{U}_{i-1}\right)\left(n^{-1}\sum_{i=1}^{n}\rho_{n}^{i-1}U_{i}\right)\nonumber \\
 & +n^{-2}\left(nQ-S^{2}\right)\left(n^{-1}\sum_{i=2}^{n}\widetilde{U}_{i-1}U_{i}\right)\Bigg],\nonumber \\
T_{2}:=\  & h_{n}\Bigg[\left(n^{-1}Qn^{-1}\sum_{i=2}^{n}\widetilde{U}_{i-1}^{2}-\left(n^{-1}\sum_{i=2}^{n}\rho_{n}^{i-1}\widetilde{U}_{i-1}\right)^{2}\right)\left(n^{-1}\sum_{i=1}^{n}U_{i}\right)\nonumber \\
 & -\left(n^{-1}Sn^{-1}\sum_{i=2}^{n}\widetilde{U}_{i-1}^{2}-\left(n^{-1}\sum_{i=2}^{n}\rho_{n}^{i-1}\widetilde{U}_{i-1}\right)\left(n^{-1}\sum_{i=2}^{n}\widetilde{U}_{i-1}\right)\right)\left(n^{-1}\sum_{i=1}^{n}\rho_{n}^{i-1}U_{i}\right)\nonumber \\
 & +n^{-1}(S-h_{n}^{1/2}Q)\left(n^{-1}\sum_{i=2}^{n}\rho_{n}^{i-1}\widetilde{U}_{i-1}\right)\left(n^{-1}\sum_{i=2}^{n}\widetilde{U}_{i-1}U_{i}\right)\Bigg],\text{ and}\nonumber \\
T_{3}:=\  & h_{n}^{1/2}\Bigg[\left(\left(n^{-1}\sum_{i=2}^{n}\rho_{n}^{i-1}\widetilde{U}_{i-1}\right)\left(n^{-1}\sum_{i=2}^{n}\widetilde{U}_{i-1}\right)-n^{-1}Sn^{-1}\sum_{i=2}^{n}\widetilde{U}_{i-1}^{2}\right)\left(n^{-1}\sum_{i=1}^{n}U_{i}\right)\nonumber \\
 & +\left(n^{-1}\sum_{i=2}^{n}\widetilde{U}_{i-1}^{2}-\left(n^{-1}\sum_{i=2}^{n}\widetilde{U}_{i-1}\right)^{2}\right)\left(n^{-1}\sum_{i=1}^{n}\rho_{n}^{i-1}U_{i}\right)\nonumber \\
 & +\left(n^{-1}Sn^{-1}\sum_{i=2}^{n}\widetilde{U}_{i-1}-n^{-1}\sum_{i=2}^{n}\rho_{n}^{i-1}\widetilde{U}_{i-1}\right)\left(n^{-1}\sum_{i=2}^{n}\widetilde{U}_{i-1}U_{i}\right)\Bigg].
\end{align}
Note that 
\begin{align}
n^{-2}(nQ-S^{2}) & =O(n^{-1}(1-\rho_{n})^{-1}),\nonumber \\
n^{-1}S & =O(n^{-1}(1-\rho_{n})^{-1}),\nonumber \\
n^{-1}Q & =O(n^{-1}(1-\rho_{n})^{-1}),\nonumber \\
n^{-1}\sum_{i=2}^{n}\widetilde{U}_{i-1}^{2} & =O_{p}((1-\rho_{n})^{-1}),\nonumber \\
n^{-1}\sum_{i=2}^{n}\rho_{n}^{i-1}\widetilde{U}_{i-1} & =o_{p}(n^{-1/2}(1-\rho_{n})^{-1}),\nonumber \\
n^{-1}(S-h_{n}^{1/2}Q) & =O(n^{-1/2}(1-\rho_{n})^{-1/2}),\nonumber \\
n^{-1}\sum_{i=2}^{n}\widetilde{U}_{i-1}U_{i} & =O_{p}(n^{-1/2}(1-\rho_{n})^{-1/2}),\nonumber \\
n^{-1}\sum_{i=1}^{n}\rho_{n}^{i-1}U_{i} & =O_{p}(n^{-1}(1-\rho_{n})^{-1/2}),\text{ and }\nonumber \\
n^{-1}\sum_{i=2}^{n}\widetilde{U}_{i-1} & =o_{p}((1-\rho_{n})^{-1/2}),
\end{align}
by parts \ref{enu:case2lem-a} and \ref{enu:case2lem-b} of Lemma
\ref{lem:case2lem}, (\ref{eq:clt-h-0}), the CLT for MDS for $h_{n}\rho_{n}^{i-1}U_{i}$,
and Lemma \ref{lem:Order-expect}. Using these results, we derive
\begin{align}
\det & =O_{p}((1-\rho_{n})^{-2}),\nonumber \\
T_{1} & =O_{p}(n^{-1/2}(1-\rho_{n})^{-1/2}),\nonumber \\
T_{2} & =O_{p}(n^{-1/2}(1-\rho_{n})^{-1}),\text{ and}\nonumber \\
T_{3} & =O_{p}(n^{-1/2}(1-\rho_{n})^{-1}).
\end{align}
Therefore, we have
\begin{equation}
(\det)^{-1}(T_{1},T_{2},T_{3})'=(O_{p}(n^{-1/2}(1-\rho_{n})^{3/2}),O_{p}(n^{-1/2}(1-\rho_{n})),O_{p}(n^{-1/2}(1-\rho_{n})))'.
\end{equation}

For part \ref{enu:case2lem-e}, note that 
\begin{align}
\widetilde{U}_{-1}'\Delta^{2}\widetilde{U}_{-1} & =\sum_{i=2}^{n}\frac{\widetilde{U}_{i-1}^{2}}{(1-p_{ii}^{*})^{2}}\left[U_{i}-(\widetilde{U}_{i-1},1,h_{n}^{1/2}\rho_{n}^{i-1})(X(\rho_{n})'X(\rho_{n}))^{-1}X(\rho_{n})'U\right]^{2}.
\end{align}
By part \ref{enu:case2lem-c} of the lemma and the fact that $p_{ii}^{*}=O_{p}(n^{-1/2})$,
it remains to be shown that 
\begin{align}
O((1-\rho_{n}))n^{-1}\sum_{i=2}^{n}\widetilde{U}_{i-1}^{2}U_{i}(\widetilde{U}_{i-1},1,h_{n}^{1/2}\rho_{n}^{i-1})(X(\rho_{n})'X(\rho_{n}))^{-1}X(\rho_{n})'U & \pto0\text{ and}\nonumber \\
O((1-\rho_{n}))n^{-1}\sum_{i=2}^{n}\widetilde{U}_{i-1}^{2}\left[(\widetilde{U}_{i-1},1,h_{n}^{1/2}\rho_{n}^{i-1})(X(\rho_{n})'X(\rho_{n}))^{-1}X(\rho_{n})'U\right]^{2} & \pto0.
\end{align}
By part \ref{enu:case2lem-d}, we only need to show that 
\begin{align}
(1-\rho_{n})^{2}n^{-3/2}\sum_{i=2}^{n}\widetilde{U}_{i-1}^{2}U_{i} & =o_{p}(1),\nonumber \\
(1-\rho_{n})^{2}n^{-2}\sum_{i=2}^{n}\widetilde{U}_{i-1}^{4} & =o_{p}(1),\nonumber \\
(1-\rho_{n})^{2}n^{-3/2}\sum_{i=2}^{n}\widetilde{U}_{i-1}^{3}U_{i} & =o_{p}(1),\nonumber \\
(1-\rho_{n})^{3}n^{-2}\sum_{i=2}^{n}\widetilde{U}_{i-1}^{2} & =o_{p}(1),\nonumber \\
(1-\rho_{n})^{2}n^{-3/2}\sum_{i=2}^{n}h_{n}^{1/2}\rho_{n}^{i-1}\widetilde{U}_{i-1}^{2}U_{i} & =o_{p}(1),\nonumber \\
(1-\rho_{n})^{3}n^{-2}\sum_{i=2}^{n}h_{n}^{1/2}\rho_{n}^{i-1}\widetilde{U}_{i-1}^{2} & =o_{p}(1),\nonumber \\
(1-\rho_{n})^{3}n^{-2}\sum_{i=2}^{n}h_{n}\rho_{n}^{2i-2}\widetilde{U}_{i-1}^{2} & =o_{p}(1),\text{ and}\nonumber \\
(1-\rho_{n})^{5/2}n^{-2}\sum_{i=2}^{n}h_{n}^{1/2}\rho_{n}^{i-1}\widetilde{U}_{i-1}^{3} & =o_{p}(1).\label{eq:boundsh4}
\end{align}
We only need to prove the last four equations in (\ref{eq:boundsh4})
because all the other results have been established in Lemma 8 of
AG12. Using the facts that 
\begin{equation}
n^{-1}\sum_{i=2}^{n}h_{n}^{1/2}\rho_{n}^{i-1}\to0\text{ and }n^{-1}\sum_{i=2}^{n}h_{n}\rho_{n}^{2i-2}\to1/2,
\end{equation}
the sixth and seventh lines are proved by part \ref{enu:case2lem-b}
of the lemma. The fifth and last lines of (\ref{eq:boundsh4}) can
be proved using Markov's inequality, as in the proof of part \ref{enu:case2lem-a},
together with the methods employed in the proof of Lemma 7 in AG12.

For part \ref{enu:case2lem-f}, we have 
\begin{align}
 & \widetilde{U}_{-1}'\Delta^{2}X_{2}(\rho_{n})\nonumber \\
=\  & \sum_{i=2}^{n}\frac{\widetilde{U}_{i-1}(1,h_{n}^{1/2}\rho_{n}^{i-1})}{(1-p_{ii}^{*})^{2}}\left[U_{i}-(\widetilde{U}_{i-1},1,h_{n}^{1/2}\rho_{n}^{i-1})(X(\rho_{n})'X(\rho_{n}))^{-1}X(\rho_{n})'U\right]^{2}.
\end{align}
Then, the desired results follow from
\begin{align}
(1-\rho_{n})^{1/2}n^{-1}\sum_{i=2}^{n}\widetilde{U}_{i-1}U_{i}^{2}(1,h_{n}\rho_{n}^{i-1}) & =(O_{p}(1),O_{p}(1))\text{ and}\nonumber \\
(1-\rho_{n})^{5/2}n^{-2}\sum_{i=2}^{n}\widetilde{U}_{i-1}\bar{U}_{i-1}^{2}(1,h_{n}\rho_{n}^{i-1}) & =(O_{p}(1),O_{p}(1)),\label{eq:lem8_part_f-g-3}
\end{align}
where $\bar{U}_{i-1}=(\widetilde{U}_{i-1},1,h_{n}^{1/2}\rho_{n}^{i-1})(X(\rho_{n})'X(\rho_{n}))^{-1}X(\rho_{n})'U$.

To prove the first line of (\ref{eq:lem8_part_f-g-3}), by Markov's
inequality we only need to show that 
\begin{align}
(1-\rho_{n})n^{-2}\E\left[(\sum_{i=2}^{n}\widetilde{U}_{i-1}U_{i}^{2})^{2}\right] & =O(1)\text{ and}\nonumber \\
(1-\rho_{n})h_{n}^{2}n^{-2}\E\left[(\sum_{i=2}^{n}\rho_{n}^{i-1}\widetilde{U}_{i-1}U_{i}^{2})^{2}\right] & =O(1).\label{eq:lem8_part_f-g-12}
\end{align}
Note that 
\begin{align}
 & \E\left(\sum_{i=2}^{n}\widetilde{U}_{i-1}U_{i}^{2}\right)^{2}\nonumber \\
=\  & \E\left(\sum_{i,j=2}^{n}\left(\sum_{k=1}^{i-1}\rho_{n}^{i-1-k}U_{k}\right)\left(\sum_{l=1}^{j-1}\rho_{n}^{j-1-l}U_{l}\right)U_{i}^{2}U_{j}^{2}\right)\nonumber \\
=\  & \sum_{i,j=2}^{n}\sum_{k=1}^{i-1}\sum_{l=1}^{j-1}\rho_{n}^{k+l}\E\left(U_{i-1-k}U_{i}^{2}U_{j-1-l}U_{j}^{2}\right)\nonumber \\
\leq\  & \sum_{i,j=2}^{n}\sum_{k=1}^{i-1}\sum_{l=1}^{j-1}\rho_{n}^{k+l}\left\{ Cov(U_{i-1-k}U_{i}^{2},U_{j-1-l}U_{j}^{2})+|\E\left(U_{i-1-k}U_{i}^{2}\right)||\E\left(U_{j-1-l}U_{j}^{2}\right)|\right\} \nonumber \\
\leq\  & M\sum_{i=j\ge2}^{n}\sum_{k=1}^{i-1}\sum_{l=1}^{j-1}\rho_{n}^{k+l}+8CM\sum_{i,j\geq2,i\neq j}^{n}\sum_{k=1}^{i-1}\sum_{l=1}^{j-1}\rho_{n}^{k+l}|i-j|^{-3-\epsilon}\nonumber \\
 & +64C^{2}M\sum_{i,j=2}^{n}\sum_{k=1}^{i-1}\sum_{l=1}^{j-1}\rho_{n}^{k+l}(k+1)^{-3-\epsilon}(l+1)^{-3-\epsilon}\nonumber \\
\leq\  & M\sum_{i=j\ge2}^{n}\left(\sum_{k=1}^{\infty}\rho_{n}^{k}\right)^{2}+8CMn\left(\sum_{i>1}^{\infty}i^{-3-\epsilon}\right)\left(\sum_{k=1}^{\infty}\rho_{n}^{k}\right)^{2}\nonumber \\
 & +64C^{2}Mn^{2}\left(\sum_{k=1}^{\infty}(k+1)^{-3-\epsilon}\right)^{2}\nonumber \\
=\  & O(n(1-\rho_{n})^{-2}+n^{2}),
\end{align}
where the second inequality uses the properties of $\alpha$-mixing
analogously to (\ref{eq:alpha-mixing-inequalities}) and the moment
conditions in parts \ref{enu:para-space3} and \ref{enu:para-space4}
of $\L_{n}$, and the last equality holds by 
\begin{align}
\sum_{i=1}^{\infty}i^{-3-\epsilon} & =O(1)\text{ and }\sum_{k=1}^{\infty}\rho_{n}^{k}=O((1-\rho_{n})^{-1}).
\end{align}
Therefore, we have 
\begin{equation}
(1-\rho_{n})n^{-2}\E\left[(\sum_{i=2}^{n}\widetilde{U}_{i-1}U_{i}^{2})^{2}\right]=O(h_{n}^{-1}+n^{-1}h_{n})=o(1),
\end{equation}
which establishes the first line of (\ref{eq:lem8_part_f-g-12}).

For the second line of (\ref{eq:lem8_part_f-g-12}), using the same
expansion and results as before, we have
\begin{align}
\E(\sum_{i=2}^{n}\rho_{n}^{i-1}\widetilde{U}_{i-1}U_{i}^{2})^{2}\leq\  & M\sum_{i=j\ge2}^{n}\rho_{n}^{i+j-2}\sum_{k=1}^{i-1}\sum_{l=1}^{j-1}\rho_{n}^{k+l}\nonumber \\
 & +8CM\sum_{i,j\geq2,i\neq j}\rho_{n}^{i+j-2}\sum_{k=1}^{i-1}\sum_{l=1}^{j-1}\rho_{n}^{k+l}|i-j|^{-3-\epsilon}\nonumber \\
 & +64C^{2}M\sum_{i,j=2}^{n}\rho_{n}^{i+j-2}\sum_{k=1}^{i-1}\sum_{l=1}^{j-1}\rho_{n}^{k+l}(k+1)^{-3-\epsilon}(l+1)^{-3-\epsilon}\nonumber \\
\leq\  & M\left(\sum_{k=1}^{\infty}\rho_{n}^{k}\right)^{2}\left(\sum_{k=1}^{\infty}\rho_{n}^{2k-2}\right)\nonumber \\
 & +8CM\left(\sum_{i,j\geq2,i\neq j}^{\infty}\rho_{n}^{i+j-2}|i-j|^{-3-\epsilon}\right)\left(\sum_{k=1}^{\infty}\rho_{n}^{k}\right)^{2}\nonumber \\
 & +64C^{2}M\left(\sum_{k=1}^{\infty}\rho_{n}^{k}\right)^{2}\left(\sum_{k=1}^{\infty}(k+1)^{-3-\epsilon}\right)^{2}\nonumber \\
=\  & O((1-\rho_{n})^{-3}),
\end{align}
where the last equality holds by $\sum_{k=1}^{\infty}\rho_{n}^{2k-2}=O((1-\rho_{n})^{-1})$
and 
\begin{align}
\sum_{i,j\geq2,i\neq j}^{\infty}\rho_{n}^{i+j-2}|i-j|^{-3-\epsilon} & =2\sum_{i>j\geq2}^{\infty}\rho_{n}^{i+j-2}(i-j)^{-3-\epsilon}=2\sum_{j=2}^{\infty}\sum_{k=1}^{\infty}\rho_{n}^{k+2j-2}k^{-3-\epsilon}\nonumber \\
 & =2\left(\sum_{j=2}^{\infty}\rho_{n}^{2j-2}\right)\left(\sum_{k=1}^{\infty}\rho_{n}^{k}k^{-3-\epsilon}\right)=O((1-\rho_{n})^{-1}).
\end{align}
Therefore, we have 
\begin{equation}
(1-\rho_{n})h_{n}^{2}n^{-2}\E\left[(\sum_{i=2}^{n}\rho_{n}^{i-1}\widetilde{U}_{i-1}U_{i}^{2})^{2}\right]=O((1-\rho_{n})^{-2}h_{n}^{2}n^{-2})=O(1).
\end{equation}

To prove the second line of (\ref{eq:lem8_part_f-g-3}), note that
\begin{align}
\bar{U}_{i-1}^{2}=\  & O_{p}(n^{-1}(1-\rho_{n})^{2})[1+h_{n}\rho_{n}^{2i-2}+(1-\rho_{n})\widetilde{U}_{i-1}^{2}\nonumber \\
 & +2h_{n}^{1/2}\rho_{n}^{i-1}+2(1-\rho_{n})^{1/2}\widetilde{U}_{i-1}+2h_{n}^{1/2}(1-\rho_{n})\rho_{n}^{i-1}\widetilde{U}_{i-1}]\label{eq:barUi-square}
\end{align}
by part \ref{enu:case2lem-d}. Following the same argument as in part
\ref{enu:case2lem-a}, we can show that 
\begin{align}
(1-\rho_{n})^{5/2}n^{-2}h_{n}^{k_{0}}\sum_{i=2}^{n}\rho_{n}^{k_{1}(i-1)}\widetilde{U}_{i-1} & =O_{p}(1)\label{eq:sum_utilde}
\end{align}
for $(k_{0},k_{1})=(0,0),(1/2,1),(1,2),(1,1),(3/2,2)$, and $(2,3)$.
For example, the case $(2,3)$ can be established by proving 
\begin{equation}
(1-\rho_{n})^{5}n^{-4}h_{n}^{4}\E\left(\sum_{i=2}^{n}\rho_{n}^{3(i-1)}\widetilde{U}_{i-1}\right)^{2}=(1-\rho_{n})^{5}n^{-4}h_{n}^{4}\times O((1-\rho_{n})^{-3})=o(1),
\end{equation}
following the same expansion as in (\ref{eq:rhontildeU-expansion})
and observing that $\sum_{j=1}^{n}\rho_{n}^{6j}=O((1-\rho_{n})^{-1})$.
Then we only need to show that 
\begin{align}
(1-\rho_{n})^{l_{0}}n^{-2}h_{n}^{l_{1}}\sum_{i=2}^{n}\rho_{n}^{l_{2}(i-1)}\widetilde{U}_{i-1}^{l_{3}} & =O_{p}(1),
\end{align}
for $(l_{0},l_{1},l_{2},l_{3})=(7/2,0,0,3)$, $(3,0,0,2)$, $(7/2,1/2,1,2)$,
$(7/2,1,1,3)$, $(3,1,1,2)$, and (7/2, 3/2,2,2).

The cases $(3,0,0,2)$ and $(7/2,1/2,1,2)$ are implied by the fourth
and sixth lines of (\ref{eq:boundsh4}), respectively.

For the case $(7/2,0,0,3)$, we have
\begin{align}
\left|(1-\rho_{n})^{7/2}n^{-2}\sum_{i=2}^{n}\widetilde{U}_{i-1}^{3}\right| & \leq(1-\rho_{n})^{2}\left((1-\rho_{n})^{2}n^{-2}\sum_{i=2}^{n}\widetilde{U}_{i-1}^{4}\right)^{1/2}\left((1-\rho_{n})n^{-2}\sum_{i=2}^{n}\widetilde{U}_{i-1}^{2}\right)^{1/2}\nonumber \\
 & =(1-\rho_{n})^{2}\times o_{p}(1)\times O_{p}(n^{-1/2})=o_{p}(1),
\end{align}
by the Cauchy-Schwarz inequality, the second line of (\ref{eq:boundsh4}),
part \ref{enu:case2lem-b}, and Lemma \ref{lem:Order-expect}.

For the case $(7/2,3/2,2,2)$, note that $\E\widetilde{U}_{i-1}^{2}=O(\E\widetilde{U}_{0}^{2})=O((1-\rho_{n})^{-1})$
by Lemma \ref{lem:Order-expect}. Hence, we have
\begin{align}
(1-\rho_{n})^{7/2}h_{n}^{3/2}n^{-2}\E\left[\sum_{i=2}^{n}\rho_{n}^{2i-2}\widetilde{U}_{i-1}^{2}\right] & =(1-\rho_{n})^{7/2}h_{n}^{3/2}n^{-2}\sum_{i=2}^{n}\rho_{n}^{2i-2}O((1-\rho_{n})^{-1}\nonumber \\
 & =O((1-\rho_{n})^{3}n^{-1/2})=o(1).\label{eq:tildeU2-bound}
\end{align}
The result follows by Markov's inequality.

For the case $(7/2,1,1,3)$, we use the Cauchy-Schwarz inequality:
\begin{align}
\  & (1-\rho_{n})^{7/2}n^{-2}h_{n}\left|\sum_{i=2}^{n}\rho_{n}^{i-1}\widetilde{U}_{i-1}^{3}\right|\nonumber \\
\leq\  & \left((1-\rho_{n})^{7/2}h_{n}^{2}n^{-2}\sum_{i=2}^{n}\rho_{n}^{2i-2}\widetilde{U}_{i-1}^{2}\right)^{1/2}\left((1-\rho_{n})^{7/2}n^{-2}\sum_{i=2}^{n}\widetilde{U}_{i-1}^{4}\right)^{1/2}\nonumber \\
=\  & O_{p}((1-\rho_{n})^{3}n^{-1/2}h_{n}^{1/2})\times o_{p}((1-\rho_{n})^{3/2})=o_{p}(1),
\end{align}
by the second line of (\ref{eq:boundsh4}) and (\ref{eq:tildeU2-bound}).
For the case $(3,1,1,2)$, by the Cauchy-Schwarz inequality, we have
\begin{align}
 & (1-\rho_{n})^{3}h_{n}\left|n^{-2}\sum_{i=2}^{n}\rho_{n}^{i-1}\widetilde{U}_{i-1}^{2}\right|\nonumber \\
\leq\  & (1-\rho_{n})h_{n}^{1/2}\left(n^{-2}(1-\rho_{n})^{3}\sum_{i=2}^{n}h_{n}\rho_{n}^{2i-2}\widetilde{U}_{i-1}^{2}\right)^{1/2}\left(n^{-2}(1-\rho_{n})\sum_{i=2}^{n}\widetilde{U}_{i-1}^{2}\right)^{1/2}\nonumber \\
=\  & (1-\rho_{n})h_{n}^{1/2}\times o_{p}(1)\times O_{p}(n^{-1/2})=o_{p}(1),
\end{align}
by the seventh line of (\ref{eq:boundsh4}), part \ref{enu:case2lem-b},
and Lemma \ref{lem:Order-expect}.

Finally, for part \ref{enu:case2lem-g}, note that 
\begin{align}
X_{2}(\rho_{n})'\Delta^{2}X_{2}(\rho_{n})= & \sum_{i=1}^{n}\frac{(1+2h_{n}^{1/2}\rho_{n}^{i-1}+h_{n}\rho_{n}^{2i-2})}{(1-p_{ii}^{*})^{2}}\left(U_{i}-\bar{U}_{i-1}\right)^{2}.
\end{align}
As $p_{ii}^{*}=O_{p}(n^{-1/2})$, we only need to establish the following:
\begin{align}
n^{-1}\sum_{i=1}^{n}\left(1+2h_{n}^{1/2}\rho_{n}^{i-1}+h_{n}\rho_{n}^{2i-2}\right)U_{i}^{2} & =O_{p}(1)\text{ and}\nonumber \\
n^{-1}\sum_{i=1}^{n}\left(1+2h_{n}^{1/2}\rho_{n}^{i-1}+h_{n}\rho_{n}^{2i-2}\right)\bar{U}_{i-1}^{2} & =O_{p}(1).\label{eq:lem8_part_f-g-6}
\end{align}
Because $\E[U_{i}^{2}]=O(1)$, the first line holds by
\begin{equation}
n^{-1}\sum_{i=1}^{n}\left(1+2h_{n}^{1/2}\rho_{n}^{i-1}+h_{n}\rho_{n}^{2i-2}\right)\leq1+2h_{n}^{1/2}n^{-1}(1-\rho_{n})^{-1}+h_{n}n^{-1}(1-\rho_{n}^{2})^{-1}=O(1).\label{eq:sum-hnrhon}
\end{equation}
By (\ref{eq:barUi-square}), it remains to show that 
\begin{align}
n^{-2}(1-\rho_{n})^{2}\sum_{i=1}^{n}\left(1+2h_{n}^{1/2}\rho_{n}^{i-1}+h_{n}\rho_{n}^{2i-2}\right)^{2} & =O(1),\nonumber \\
n^{-2}(1-\rho_{n})^{3}\sum_{i=2}^{n}\left(1+2h_{n}^{1/2}\rho_{n}^{i-1}+h_{n}\rho_{n}^{2i-2}\right)\widetilde{U}_{i-1}^{2} & =O_{p}(1),\text{ and}\nonumber \\
n^{-2}(1-\rho_{n})^{5/2}\sum_{i=2}^{n}\left(1+2h_{n}^{1/2}\rho_{n}^{i-1}+h_{n}\rho_{n}^{2i-2}\right)\left(1+h_{n}^{1/2}(1-\rho_{n})^{1/2}\rho_{n}^{i-1}\right)\widetilde{U}_{i-1} & =O_{p}(1).
\end{align}
The first equality holds by (\ref{eq:sum-hnrhon}) and the fact that
\begin{equation}
n^{-2}(1-\rho_{n})^{2}h_{n}^{3/2}\sum_{i=1}^{n}\rho_{n}^{3i-3}=o(1)\text{ and }n^{-2}(1-\rho_{n})^{2}h_{n}^{2}\sum_{i=1}^{n}\rho_{n}^{4i-4}=o(1).
\end{equation}
By Markov's inequality, the second equality holds by (\ref{eq:sum-hnrhon})
and the fact $\E\widetilde{U}_{i-1}^{2}=O((1-\rho_{n})^{-1})$. Finally,
the last equality holds by part \ref{enu:case2lem-a} and (\ref{eq:sum_utilde})
for cases $(1,2)$ and $(2,3)$.
\end{proof}

\section{Additional Simulation Results}\label{sec:Additional-Simulation-Results}\label{SM-sec:Additional-Simulation-Results}

We provide additional simulation results in this section. Table \ref{tab:CP-mik07}
reports the CPs of the \textcite[Mik07 hereafter]{mikusheva2007uniform}
CI.\footnote{The Mik07 CI differs from the \textcite{andrews2014conditional} CI
in that there is no initial-condition component in the asymptotic
distribution of the test statistic, and a homoskedastic variance estimator
is used.} The CPs of the Mik07 CI are close to the nominal level when the $U_{i}$
sequences are i.i.d. and $Y_{0}^{*}$ is either fixed or stationary.
When $Y_{0}^{*}$ is drawn from scaled $n$ or explosive distributions,
or when the $U_{i}$ sequences are non-i.i.d., the CPs of the Mik07
CI fall below 95\%.

\begin{table}[H]
\centering
\caption{Coverage probabilities ($\times100$) of the nominal 95\% Mik07 CI}\label{tab:CP-mik07}

\begin{tabular}{lrrrrrlrrrrr}
\toprule
\multicolumn{1}{r}{\small\textbf{Initial Conditions:}} & \multicolumn{5}{c}{\small\textbf{Fixed}} & & \multicolumn{5}{c}{\small\textbf{Stationary}} \\
\multicolumn{1}{r}{\small\textbf{$\boldsymbol{\rho}$:}} & .00 & .50 & .70 & .90 & .99 & & .00 & .50 & .70 & .90 & .99 \\
\midrule
\small\textbf{i.i.d.} & 94.9 & 94.7 & 94.9 & 94.8 & 94.7 & & 94.8 & 94.7 & 94.9 & 94.7 & 94.0 \\
\small\textbf{GARCH1} & 93.8 & 93.8 & 94.3 & 94.4 & 94.5 & & 93.8 & 93.8 & 94.2 & 94.4 & 93.6 \\
\small\textbf{GARCH2} & 89.6 & 90.6 & 90.7 & 92.3 & 93.4 & & 89.6 & 90.7 & 90.6 & 92.4 & 92.8 \\
\small\textbf{GARCH3} & 85.3 & 86.0 & 87.3 & 89.4 & 92.3 & & 85.3 & 86.0 & 87.1 & 89.6 & 91.8 \\
\small\textbf{ARCH4} & 80.3 & 81.6 & 82.5 & 86.3 & 91.1 & & 80.4 & 81.6 & 82.7 & 86.5 & 90.8 \\
\midrule
\small\textbf{Initial Conditions:} & \multicolumn{5}{c}{\small\textbf{Scaled $\boldsymbol{n}$}} & & \multicolumn{5}{c}{\small\textbf{Explosive}} \\
\midrule
\small\textbf{i.i.d.} & 94.8 & 94.6 & 94.7 & 93.5 & 80.4 & & 94.6 & 94.4 & 94.2 & 92.3 & 73.5 \\
\small\textbf{GARCH1} & 94.3 & 93.8 & 93.7 & 92.7 & 79.8 & & 94.3 & 93.8 & 93.6 & 91.5 & 73.0 \\
\small\textbf{GARCH2} & 91.1 & 91.9 & 91.7 & 92.2 & 80.3 & & 92.4 & 92.8 & 92.5 & 91.4 & 73.0 \\
\small\textbf{GARCH3} & 87.5 & 88.5 & 89.2 & 90.8 & 80.2 & & 90.2 & 90.7 & 91.2 & 90.9 & 73.4 \\
\small\textbf{ARCH4} & 84.0 & 84.8 & 86.5 & 89.2 & 80.4 & & 88.5 & 89.0 & 89.7 & 90.3 & 73.6 \\
\bottomrule
\end{tabular}
\end{table}

Table \ref{tab:length-ratio-mik07} reports the ratios of the ALs
of ICR to Mik07 CIs when the $U_{i}$ sequences are i.i.d. and $Y_{0}^{*}$
is fixed or stationary. The ratios are between 1 and 1.12, suggesting
that the ALs of the ICR CI are similar to those of the Mik07 CI when
the true data generating process is such that the Mik07 inference
method is valid. Tables \ref{tab:avg_length_ag14} and \ref{tab:avg_length_mik07}
report the ALs of \textcite[AG14 hereafter]{andrews2014conditional}
and the Mik07 CI, respectively. The results are consistent with the
conclusions drawn from Tables \ref{tab:length-ratio} and \ref{tab:length-ratio-mik07}.

\begin{table}[h]
\centering
\caption{Ratios of the average lengths of the nominal 95\% ICR to Mik07 CIs}\label{tab:length-ratio-mik07}

\begin{tabular}{lrrrrrlrrrrr}
\toprule
\multicolumn{1}{r}{\small\textbf{Initial Conditions:}} & \multicolumn{5}{c}{\small\textbf{Fixed}} & & \multicolumn{5}{c}{\small\textbf{Stationary}} \\
\multicolumn{1}{r}{\small\textbf{$\boldsymbol{\rho}$:}} & .00 & .50 & .70 & .90 & .99 & & .00 & .50 & .70 & .90 & .99 \\
\midrule
{\small\textbf{i.i.d.}} & 1.00 & 1.01 & 1.02 & 1.03 & 1.09 & & 1.01 & 1.01 & 1.03 & 1.05 & 1.12 \\
\bottomrule
\end{tabular}
\end{table}

\begin{table}[h]
\centering
\caption{Average lengths of the nominal 95\% AG14 CI }\label{tab:avg_length_ag14}

\begin{tabular}{lrrrrrlrrrrr}
\toprule
\multicolumn{1}{r}{\small\textbf{Initial Conditions:}} & \multicolumn{5}{c}{\small\textbf{Fixed}} & & \multicolumn{5}{c}{\small\textbf{Stationary}} \\
\multicolumn{1}{r}{\small\textbf{$\boldsymbol{\rho}$:}} & .00 & .50 & .70 & .90 & .99 & & .00 & .50 & .70 & .90 & .99 \\
\midrule
\small\textbf{i.i.d.} & .32 & .28 & .24 & .16 & .07 & & .32 & .28 & .24 & .16 & .07 \\
\small\textbf{GARCH1} & .32 & .28 & .24 & .16 & .07 & & .33 & .29 & .25 & .17 & .07 \\
\small\textbf{GARCH2} & .38 & .34 & .28 & .18 & .08 & & .38 & .33 & .28 & .18 & .07 \\
\small\textbf{GARCH3} & .43 & .38 & .32 & .19 & .08 & & .43 & .38 & .32 & .19 & .07 \\
\small\textbf{ARCH4} & .49 & .43 & .36 & .21 & .09 & & .48 & .43 & .36 & .20 & .08 \\
\midrule
\small\textbf{Initial Conditions:} & \multicolumn{5}{c}{\small\textbf{Scaled $\boldsymbol{n}$}} & & \multicolumn{5}{c}{\small\textbf{Explosive}} \\
\midrule
\small\textbf{i.i.d.} & .22 & .19 & .16 & .09 & .02 & & .10 & .09 & .08 & .04 & .01 \\
\small\textbf{GARCH1} & .05 & .06 & .07 & .05 & .02 & & .00 & .01 & .02 & .02 & .01 \\
\small\textbf{GARCH2} & .27 & .23 & .19 & .10 & .02 & & .13 & .12 & .09 & .04 & .01 \\
\small\textbf{GARCH3} & .31 & .27 & .21 & .11 & .02 & & .15 & .14 & .11 & .05 & .01 \\
\small\textbf{ARCH4} & .35 & .30 & .24 & .12 & .03 & & .17 & .15 & .12 & .05 & .01 \\
\bottomrule
\end{tabular}
\end{table}

\begin{table}[h]
\centering
\caption{Average lengths of the nominal 95\% Mik07 CI}\label{tab:avg_length_mik07}

\begin{tabular}{lrrrrrlrrrrr}
\toprule
\multicolumn{1}{r}{\small\textbf{Initial Conditions:}} & \multicolumn{5}{c}{\small\textbf{Fixed}} & & \multicolumn{5}{c}{\small\textbf{Stationary}} \\
\multicolumn{1}{r}{\small\textbf{$\boldsymbol{\rho}$:}} & .00 & .50 & .70 & .90 & .99 & & .00 & .50 & .70 & .90 & .99 \\
\midrule
\small\textbf{i.i.d.} & .32 & .28 & .24 & .17 & .07 & & .32 & .28 & .24 & .16 & .07 \\
\small\textbf{GARCH1} & .31 & .27 & .23 & .16 & .07 & & .32 & .28 & .24 & .16 & .07 \\
\small\textbf{GARCH2} & .32 & .28 & .24 & .17 & .07 & & .32 & .28 & .24 & .16 & .07 \\
\small\textbf{GARCH3} & .32 & .28 & .24 & .17 & .08 & & .32 & .28 & .24 & .16 & .07 \\
\small\textbf{ARCH4} & .32 & .28 & .24 & .17 & .08 & & .32 & .28 & .24 & .16 & .07 \\
\midrule
\small\textbf{Initial Conditions:} & \multicolumn{5}{c}{\small\textbf{Scaled $\boldsymbol{n}$}} & & \multicolumn{5}{c}{\small\textbf{Explosive}} \\
\midrule
\small\textbf{i.i.d.} & .24 & .20 & .16 & .09 & .02 & & .14 & .11 & .08 & .04 & .01 \\
\small\textbf{GARCH1} & .10 & .09 & .08 & .05 & .02 & & .02 & .02 & .02 & .02 & .01 \\
\small\textbf{GARCH2} & .25 & .21 & .17 & .09 & .02 & & .15 & .12 & .09 & .04 & .01 \\
\small\textbf{GARCH3} & .25 & .21 & .17 & .10 & .02 & & .15 & .12 & .09 & .05 & .01 \\
\small\textbf{ARCH4} & .25 & .21 & .17 & .10 & .02 & & .15 & .12 & .10 & .05 & .01 \\
\bottomrule
\end{tabular}
\end{table}

\clearpage{}

\newpage{}

\printbibliography[title={\centering\makebox[\textwidth]{References for the Supplemental Material}}]

\end{refsection}
\end{document}